\documentclass[letterpaper,11pt]{article}

\usepackage{amsmath}
\usepackage{amsfonts}
\usepackage{amssymb}
\usepackage{amsthm}
\usepackage{thmtools}

\declaretheorem[name=Theorem]{theorem}
\declaretheorem[name=Corollary, sibling=theorem]{corollary}
\declaretheorem[name=Definition, sibling=theorem]{definition}
\declaretheorem[name=Lemma, sibling=theorem]{lemma}

\usepackage{graphicx}
\usepackage{xspace}
\usepackage{paralist}
\usepackage[margin=1in]{geometry}
\usepackage[algo2e,noend,ruled,linesnumbered]{algorithm2e}

\ifx\pdftexversion\undefined
\usepackage[colorlinks,linkcolor=black,filecolor=black,citecolor=black,urlcolor=black,pdfstartview=FitH]{hyperref}
\else
 \usepackage[colorlinks,linkcolor=blue,filecolor=blue,citecolor=blue,urlcolor=blue,pdfstartview=FitH]{hyperref}
\fi

\usepackage{cleveref}

\newcommand{\N}{\mathbb{N}}
\newcommand{\R}{\mathbb{R}}
\newcommand{\Rp}{\R_{\ge 0}}
\newcommand{\BO}{\mathcal{O}}
\newcommand{\G}{\mathcal{G}}
\newcommand{\Ll}{\mathcal{L}}
\newcommand{\T}{\mathcal{T}}
\newcommand{\A}{\mathcal{A}}
\newcommand{\Wdelta}{\mathcal{W}_{\delta}}
\newcommand{\W}{\mathcal{W}}
\newcommand{\Wd}{\mathcal{W}_d}

\begin{document}

\title{Gradient Clock Synchronization\\
with Practically Constant Local Skew\bigskip}

\author{Christoph Lenzen\\
CISPA Helmholtz Center for Information Security, Saarbr\"ucken, Germany\\
and Reykjavik University, Reykjavik, Iceland}

\date{}

\maketitle

\begin{abstract}
Gradient Clock Synchronization (GCS) is the task of minimizing the \emph{local skew,} i.e., the clock offset between neighboring clocks, in a larger network.
While asymptotically optimal bounds are known, from a practical perspective they have crucial shortcomings:
\begin{itemize}
  \item Local skew bounds are determined by upper bounds on offset estimation that need to be guaranteed throughout the entire lifetime of the system.
  \item Worst-case frequency deviations of local oscillators from their nominal rate are assumed, yet frequencies tend to be much more stable in the (relevant) short term.
\end{itemize}
State-of-the-art deployed synchronization methods adapt to the true offset measurement and frequency errors, but achieve no non-trivial guarantees on the local skew.

In this work, we provide a refined model and novel analysis of existing techniques for solving GCS in this model.
By requiring only \emph{stability} of measurement and frequency errors, we can circumvent existing lower bounds, leading to dramatic improvements under very general conditions.
For example, if links exhibit a uniform worst-case estimation error of $\Delta$ and a \emph{change} in estimation errors of $\delta\ll \Delta$ on relevant time scales, we bound the local skew by $\BO(\Delta+\delta \log D)$ for networks of diameter $D$, effectively ``breaking'' the established $\Omega(\Delta\log D)$ lower bound, which holds when $\delta=\Delta$.
Similarly, we show how to limit the influence of local oscillators on $\delta$ to scale with the \emph{change} of frequency of an individual oscillator on relevant time scales, rather than a worst-case bound over all oscillators and the lifetime of the system.
Surprisingly, these results require only very limited knowledge of $\Delta$.
In particular, the likely dominant term of $\BO(\Delta)$ in the local skew is determined by the \emph{actual} link performance, not an upper bound covering worst-case conditions.

Moreover, we show how to ensure self-stabilization in this challenging setting.
Last, but not least, we extend all of our results to the scenario of external synchronization, at the cost of a limited increase in stabilization time.
\end{abstract}
\medskip

\section{Introduction and Related Work}\label{sec:intro}

Clock synchronization, the task of minimizing the offset---also known as \emph{skew}---between clocks in a distributed network, has been studied extensively.
Being a distributed task by nature, it drew the attention of researchers in distributed computing early on, see e.g.~\cite{attiya96,dolev95,upperlowerbound,driftingclocks}.
In particular, Biaz and Welch established that in a network of diameter $D$, the worst-case skew that can be guaranteed between all pairs of nodes---called \emph{global skew}---grows linearly with~$D$~\cite{biaz01closed}.

One of the main benefits of matching lower bound constructions is that they do not only prove what the ``best'' bound is, but also tend to shed light on the true obstacle to better performance.
In this case, the issue is that nodes at a great distance from each other have no way of getting accurate information on each other's current clock value.
However, many applications require only the clocks of nearby nodes to be well-synchronized.
In their seminal work from 2004~\cite{fan04gcs}, Fan and Lynch focused on this aspect, introducing the problem of Gradient Clock Synchronization (GCS).
GCS asks to minimize the worst-case skew between neighbors, on which the lower bound construction from~\cite{biaz01closed} has no non-trivial implication.
The key insight provided by Fan and Lynch was that, nonetheless, this so-called \emph{local skew} must grow as $\Omega(\log D/\log \log D)$.
Fan and Lynch showed this based on (unknown) variations in message delivery and processing times, but essentially the same argument can be made based on (unknown) frequency variations in local oscillators~\cite{gradientclocksensor}.

Follow-up work pinpointed the achievable local skew to $\Theta(\log D)$~\cite{lenzen2010tight}, extended the results to non-uniform~\cite{kuhn09reference} and dynamic networks~\cite{kuhn10dynamic,kuhn09}, and added tolerance to Byzantine faults~\cite{bund2019fault}.
Moreover, it has been demonstrated that hardware implementations of the algorithm are feasible and promising~\cite{bund2020pals,bund2023pals}, strongly supporting that practical application of the algorithmic technique is of interest.
However, all of these works assume a known (and thus worst-case!) bound $\Delta$ on the clock offset estimation error across each edge, and let node $v$ increase its clock value to match the largest one it perceives in its neighborhood, so long as this does not cause a perceived offset of more than $\Delta$ to any of its neighbors.
This results in skews of at least $\Delta$, even if the \emph{actual} measurement error is much smaller.

In contrast, ``naively'' synchronizing via a BFS tree rooted at $r$ by tracking the parent's clock value would result in a global skew of at most $2\max_{v\in V}\{|d(r,v)|\}$, where $d(r,v)$ is the (possibly negative) distance induced by labeling each tree edge with the actual measurement error.
Loosely speaking, this is in fact what deployed protocols like the Network Time Protocol (NTP)~\cite{mills91} and the Precision Time Protocol (PTP)~\cite{ptp} \emph{do,} typically resulting in significantly better performance than the aforementioned lower bounds would suggest.
The reason is that a worst-case bound $\Delta$ on measurement errors that needs to hold across the entire lifetime of all systems deploying the algorithm is extremely conservative.

At first glance, the lower bound constructions from~\cite{fan04gcs,lenzen2010tight} seem to show that this price needs to be paid:
as with all indistinguishability arguments, it is the \emph{possibility} that the errors might be large that forces a worst-case optimal algorithm to account for this scenario, even if this contingency is not realized.
Fortunately, a closer inspection of the lower bound constructions reveals that this is not so!
A critical step in these constructions is to suddenly change offset estimation errors as much as possible, swinging by $\Delta$ on some edges.
This provides fresh ``slack'' for building up even larger local skews without the algorithm being able to perceive where this takes place.
However, in many settings, the amount by which the measurement error can fluctuate on the time scale relevant for the lower bound is much smaller.

Similarly, the lower bound construction introduces additional skew by modifying the speed of local oscillators.
Assuming that these oscillators run at a rate of at least $1$ and at most $\vartheta>1$, skew can be introduced into the system at a rate of $\vartheta-1$.
Again, frequency varies much more over the system's lifetime than on the time scale of the construction.
Not coincidentally, practically deployed protocols compensate for local oscillators' frequency errors by locking their frequency to that of an (external) time reference~\cite{ptp,mills91}.
This also substantially improves performance in practice.
Accordingly, the main question motivating this work is the following.
\begin{quote}
\textit{Is it possible to exploit \emph{stability} of measurement errors and clock frequencies to achieve stronger bounds on the local skew?}
\end{quote}

\begin{table}
\caption{Notation cheat sheet.\label{tab:cheat}}\smallskip
\begin{tabular}{c|c|p{8.62cm}}
variable & meaning & comments\\
\hline
$(V,E)$ & network graph & simple; assumed to be given\\
$D$ & network diameter & unweighted\\
$r$ & designated root & used to control self-stabilization mechanism\\
\hline
$L_v(t)$ & output clock of $v$ & also referred to as logical clock\\
$\G(t)$ & global skew & $\max_{v,w\in V}\{L_v(t)-L_w(t)\}$\\
$\Ll_e(t)$ & local skew on edge $e$ & $|L_v(t)-L_w(t)|$, where $e=\{v,w\}$\\
$\Ll(t)$ & local skew & $\max_{e\in E}\{\Ll_e(t)\}$\\
$\T(t)$ & real-time skew & $\max_{v\in V}\{|L_v(t)-t|\}$\\
\hline
$o_{v,w}(t)$ & $v$'s offset estimate for $w$ & for $\{v,w\}\in E$; $o_{v,w}(t)\approx L_v(t)-L_w(t)$\\
$e_{v,w}(t)$ & error of $o_{v,w}(t)$ & $e_{v,w}(t)=L_v(t)-L_w(t)-o_{v,w}(t)$; $e_{v,w}(t)\approx -e_{w,v}(t)$\\
$\Delta$ & worst-case bound for $o_{v,w}$ & $\max_{\{v,w\}\in E}\{|e_{v,w}(t)|\}\le \Delta$\\
$\delta_{\{v,w\}}(T)$ & change of $o_{v,w}$ over $T$ time & for $|t'-t|\le T$, $|e_{v,w}(t')-e_{v,w}(t)|< \delta_{\{v,w\}}(T)$; $T$ is fixed and hence typically omitted from notation\\
$\delta(T)$ & upper bound on $\delta_{\{v,w\}}(T)$ & special case of uniform speed of changes in errors\\
\hline
$H_v(t)$ & local oscillator of $v$ & also referred to as hardware clock\\
$\vartheta$ & maximum rate of $H_v$ & given; $1\le \frac{dH_v}{dt}(t)\le \vartheta$\\
$\mu$ & speedup factor for $L_v$ & chosen; usually $\frac{dH_v}{dt}(t)\le \frac{dL_v}{dt}(t)\le (1+\mu)\frac{dH_v}{dt}(t)$\\
$\alpha, \beta$ & rate bounds for $L_v$ & $\alpha \le \frac{dL_v}{dt}(t)\le \beta$\\
$S$ & stabilization time & time until skew bounds hold; we ensure $S\in \BO(T)$\\
$\sigma$ & base of log in bound for $\Ll$ & $\sigma=\mu/(\vartheta-1)$ or $\sigma=\mu/(\zeta-1)$ for internal and external synchronization, respectively\\
\hline
$d_e$ & max.\ communication delay & given; we assume $\delta_e\in \Omega((\beta-\alpha)d_e)$\\
$\W_d$ & flooding time & weighted diameter of $(V,E,d_e)$\\
\hline
$R\subseteq V$ & nodes with ext.\ reference & given; $R=\emptyset$ for internal synchronization\\
$\zeta$ & ext.\ sync.\ slowdown factor & chosen; $1+ 2(\vartheta-1)\le \zeta\le 1+\mu$; trades between stabilization time and skew bound in external sync.\\
$\delta_v(T)$ & external ref.\ error change & as $\delta_{\{v,w\}}(T)$, but w.r.t.\ the error of $v$'s estimate of $t$\\
$H$ & virtual graph & created by adding virtual $v_0$, connected to all $v\in R$\\
$D_H$ & diameter of $H$ & unweighted; $1\le D_H\le D+1$\\
\hline
$P$ & input period of PLL & PLLs expect reference signal of roughly this period\\
$\nu(P)$ & local ref.\ stability & with perfect input, PLL would be this accurate\\
\hline
$s$ & level & alg.\ described by triggers and conditions for $s\in \N$\\
$[a,a+T]$ & analyzed time interval & fixed throughout analysis; general statements follow by covering $\R_{\ge 0}$ with overlapping intervals\\
$O_{v,w}$ & nominal offset & $O_{v,w}=(e_{v,w}(a+T/2)-e_{w,v}(a+T/2))/2$\\
\hline
&&\\[-2.5ex]
$\vec{E}$ & oriented edges & $\vec{E}:=\bigcup_{\{v,w\}\in E}\{(v,w),(w,v)\}$\\
$\omega^s(v,w)$ & level-$s$ weights & $\omega^s(v,w):=4s\delta_{\{v,w\}}-O_{v,w}$ for $(v,w)\in \vec{E}$, $s\in 1/2\cdot \N$\\
$\omega^s(W)$ & level-$s$ weight of walk $W$ & $\sum_{i=1}^{\ell}\omega^s(v_{i-1},v_i)$, where $W=v_0,v_1,\ldots,v_{\ell}$\\
$G^s$ & level-$s$ graph & $G^s=(V,\vec{E},\omega^s)$\\
$s_0$ & max.\ level with neg.\ cycles & no negative cycle in $G_s$ for $s>s_0$\\
$d^s(v,w)$ & distance in $G^s$ & well-defined for $s>s_0$\\
$\W^s$ & diameter of $G^s$ & weighted; $\W^s=\max_{v,w\in V}\{d^s(v,w)\}$\\
$\Psi_v^s(t)$ & level-$s$ potential at $v$ & $\Psi_v^s(t)=\max_{w\in V}\{L_w(t)-L_v(t)-d^s(v,w)\}$ for $s>s_0$\\
$\Psi^s(t)$ & level-$s$ potential & $\Psi^s(t)=\max_{v\in V}\{\Psi_v^s(t)\}$ for $s>s_0$
\end{tabular}
\end{table}
\subsection*{Main Results}
We provide a positive answer to the above question.
Introducing the model assumption that measurement errors on each edge do not change by more than $\delta(T)\ll \Delta$ within $T$ time, cf.~\Cref{tab:cheat}, the lower bound from~\cite{lenzen2010tight} becomes $\Omega(\delta(T)\log D)$ with $T$ being the time to build up this amount of skew using the construction.
We provide a matching upper bound.
\begin{corollary}\label{cor:uniform}
Suppose that $\mu>2\vartheta-1$ and $T\ge C \Delta D/\mu$ for a sufficiently large constant $C>0$, and denote $\sigma:=\mu/(\vartheta-1)$.
Assuming that output clocks can be initialized such that $\G(0)=\BO(D\Delta)$, there is a GCS algorithm outputting clocks running at rates between $1$ and $\vartheta(1+\mu)$ that at times $t\ge T$ achieves global and local skew bounds of
\begin{align*}
\G(t)\in \left(1+\frac{3}{\sigma-1}\right)(\Delta+\BO(\delta(T)))D \quad \mbox{and} \quad\Ll(t)\in 3\Delta+4\delta(T)(\log_{\sigma}D+\BO(1)).
\end{align*}
\end{corollary}
Note that $\delta(T)$ is increasing in $T$, i.e., it is desirable to choose $T$ small.
However, $T$ needs to be large enough for the algorithm to adjust clocks to (perceived) skews.
This results in the lower bound on $T$, which is proportional to (the bound on) $\G(t)$ divided by $\mu$, the parameter of the algorithm controlling by how much output clocks may run faster than local oscillators.
We remark that (i) one could further reduce the lower bound on $T$ by considering by how quickly the dynamics of link performance \emph{change} the (perceived) global skew within $T$ time and (ii) the $O$-notation does not hide large constants.
For the sake of a more streamlined exposition, we opted against covering these optimizations.
Likewise, we refrain from modeling how to initialize the clocks, which for the required weak bound of $\G(0)=\BO(D\Delta)$ is straightforward in any reasonable model.

It is worth noting that in contrast to the self-stabilizing variant of the algorithm discussed below, \Cref{cor:uniform} and the more general results presented in \Cref{sec:algo} are maximally communication-efficient: by working exclusively based on estimates of clock offsets to neighbors, they introduce no overhead in communication over solving the task to compute local, sufficiently accurate such estimates.\footnote{While this abstraction prevents us from quantifying the number and size of messages, the reader is invited to view this as a feature, not a bug. Deriving these estimates can be highly specific to the context. Our results offer a clean separation between the concerns of how to locally estimate offsets and how to use these measurements to adjust the clocks to obtain the best global and local skew bounds.}

\paragraph*{Adaptivity w.r.t.\ Measurement Errors.}
Observe that even for relatively small values of $\mu$, one is unlikely to encounter a network large enough for the logarithmic term to dominate the local skew bound, unless $\delta$ is fairly close to $\Delta$.
However, the above statement is still expressed in terms of $\Delta$, a worst-case bound on offset measurements.
The bounds we show are actually stronger.
Denoting for edge $\{v,w\}$ by $e_{v,w}(t)$ the error in the measurement $v$ has of the offset between its own and $w$'s output clock at time $t$, the local skew is bounded by
\begin{equation*}
|e_{v,w}(t)|+\delta(T)\left(4s+\BO\left(\log_{\sigma}\frac{\W^s}{\delta(T)}\right)\right),
\end{equation*}
where $s\in \N$ is minimal such that the network graph with edge weights $4(s-1/2)\delta(T)-e_{v,w}(t)$ has no negative cycles and $\W^s$ is its diameter w.r.t.\ the induced distances.
Note that if $\delta(T)$ is small, $s$ will be such that the $4(s-1/2)\delta(T)$ term is just enough to match the sum $\Sigma$ of the $e_{v,w}(t)$ values on the path maximizing the sum of errors $e_{v,w}(t)$ along its edges.
If $e_{v,w}(t)\approx -e_{w,v}(t)$, which we will ensure, this means that $\W^s\approx 2\Sigma$.
Intuitively, this means that the obtained bound is ``independent'' of $\delta(T)$, so long as the logarithmic term is not dominant.
Therefore, in this regime the algorithm is fully adaptive w.r.t.\ to the quality of offset estimation.
For the general case, we observe that $e_{v,w}(t)\le \Delta$, implying $s\le \lceil 1/2+\Delta/(4\delta(T))\rceil$ and $\W^s\le (\Delta+6\delta)D$, giving rise to \Cref{cor:uniform}.
The situation is similar for the global skew bound.

Without additional assumptions, each of the terms in the above refined local skew bound is necessary.
The logarithmic term corresponds to the known lower bound on the local skew from~\cite{lenzen2010tight}.
To see that the other two terms are essentially tight, consider a cycle in which all measurements are perfectly accurate, except for a single edge $e$, where the error is always $f$ in one direction and $-f$ in the other. We now initialize clock values so that all perceived offsets are identical, regardless of which edge is inaccurate.\footnote{This is valid when we consider self-stabilization, i.e., the initial state of the system can be chosen adversarially. Otherwise, start with all clock values equal and errors of $0$, and slowly build up skew and adjust errors to maintain indistinguishability. Given enough time, this can be done with arbitrarily small $\delta(T)$.}
While an algorithm can adjust clock values, indistinguishability means that the end result in terms of \emph{measured} skews is identical, regardless of our choice of $e=\{v,w\}$.
Thus, the best an algorithm can do is to balance the observed offsets between neighbors, which along the $n$-node cycle sum to $f$ or $-f$, respectively, depending on direction, yielding that the algorithm adjusts clocks so that observed skews are $f/n$ resp.\ $-f/n$.

Up to sign, this results in a skew of $(n-1)f/n$ between $v$ and $w$ and $f/n$ between any other pair of neighbors.
Moreover, the cycle has weight $4n(s-1/2)\delta(T)\pm f$ (depending on direction), so $s\approx f/(4n\delta(T))$.
Here, we can make $\delta(T)$ arbitrarily small, as we only need to change measurement errors to maintain indistinguishability while building up skew.
In particular, choosing $\delta(T)$ sufficiently small ensures that $4s\delta(T)\approx f/n$, $\W^s\approx 2f$, and $\delta(T)\log_{\sigma}(\W^s/\delta(T))\approx \delta(T)\log_{\sigma}(2f/\delta(T))\ll f/n$.

We conclude that for any pair of neighbors $v'$, $w'$ different from $v$, $w$, due to $e_{v',w'}(t)=0$ the term of $4s\delta(T)$ is necessary to account for the unavoidable skew in this scenario.
Likewise, for the skew between $v$ and $w$, the term $|e_{v,w}(t)|=f$ is essential;
the relative contribution of the other terms becomes arbitrarily small when $n$ is made large and $\delta(T)$ small.

Of course, it is worth emphasizing that we could avoid the large skew across $e$ and, in fact, elsewhere, if we \emph{knew} that the alternative path from $v$ to $w$ offers more precise offset estimates.
However, this would require \emph{known} bounds on the estimation errors that allow to maintain smaller skews.
In contrast, we see that it is impossible to bound the skew on edge $\{v,w\}$ solely in terms of $e_{v,w}$ and some notion of distance pertaining to actual measurement errors.

In analogy to~\cite{kuhn09reference}, all of our results hold also for heterogeneous graphs, i.e., $\delta(T)$ can be different for each edge $e$.
However, for the sake of exposition and easier comparison to related work, we confine ourselves to the uniform case in this overview of our results.
Moreover, we stress that it is not possible to completely uncouple the algorithm from \emph{some} knowledge of $\W^s$ and, by extension, $\Delta$.
The parameter $\delta(T)$ needs to be known to the algorithm, and the condition $T\ge C \W^s/\mu$ must hold.
Therefore, $\delta(T)$ can only be chosen correctly based on an upper bound on $\W^s$.
However, typically, a safe bound for $\W^s$ will be substantially smaller than $D \Delta$, and $\delta(T)$ is likely small enough for other terms to be dominant even when $T$ is chosen conservatively.

\paragraph*{Self-Stabilization.}
Self-stabilization~\cite{dijkstra74self}, i.e., recovery of the system from arbitrary transient faults, is equivalent to convergence to in-spec operation from arbitrary initial states.
To be of practical utility, this should occur within a small \emph{stabilization time}.
Arguably, it is implicit in previous work~\cite{kuhn10dynamic_arxiv} that existing GCS algorithms can be made self-stabilizing.
Given the results in~\cite{kuhn10dynamic_arxiv}, all that needs to be added is a detection mechanism for excessive global skew, triggering a reset of the output clocks if necessary.
The situation is similar in this work, but devising such a detection mechanism is challenging.
Our second main result is that such a detection can be realized efficiently, i.e., our GCS algorithm can indeed be made self-stabilizing.
\begin{corollary}\label{cor:stab}
There is a self-stabilizing algorithm with stabilization time $\BO(\Delta D/\mu)$ achieving the same guarantees as stated in \Cref{cor:uniform}.
\end{corollary}
Put simply, the stabilization time asymptotically matches the bound on the global skew divided by $\mu$, i.e., the time necessary for the algorithm to reduce skews provided that they are not clearly so large that transient faults must have occurred recently.

We remark that this result requires that (multi-hop) communication between all nodes is possible within $\BO(\Delta D/\mu)$ time.
This is the case if such communication uses the same channels as offset estimation and bandwidth is not a constraint, as an end-to-end communication delay of $d$ on each edge implies that $\Delta\in \Omega(\mu d)$.
Otherwise, the delay for implementing this communication enters the above bound on the stabilization time.
In this work, we make no effort at controlling the size of messages.\footnote{Large messages are required only to check whether the global skew bound is violated substantially. When collecting the data required for this non-local check, neither high timing accuracy nor low latency is required for the respective messages. For example, in a ``wireless'' setting it would be essential to use the wireless channel for synchronoization purposes, but a Local Area Network or Internet connection unsuitable for accurate synchronization could be used to collect the monitoring data required for the stabilization mechanism.}
However, we conjecture that a Bellman-Ford-style computation can achieve the same result with small messages.
We leave exploring this to future work.

Another nuance of this result is that estimating the global skew up to an error of $\BO(\W^s)$ remains elusive.
Instead, the estimation routine achieves this only with error $\BO(\W^{\tilde{s}})$ for some $s\le \tilde{s}\le s+\BO(1)$.
While this has no asymptotic effect on the stabilization time in \Cref{cor:stab}, in general it is possible that $\W^{\tilde{s}}\gg \W^s$.
However, this should be considered an artifact of the possibility that $\W^s$ can be ``too small.''
Since $\delta(T)$ needs to be a conservative bound, it is likely to be larger than necessary for typical executions.
However, increasing $\delta(T)$ even by a constant factor can change $\W^s$ in a similar way to replacing $s$ by $\tilde{s}$.

\paragraph*{External Synchronization.}
In many application scenarios, it is desired that the system time is not only internally consistent, but tied to an external reference, e.g.\ UTC.
Previous work provides a mechanism for this~\cite[Section~8.5]{lenzen2010tight}, but careful reasoning is required to show that it can be made compatible with our novel analysis of the employed algorithmic techniques.

Our solution is almost seamless, up to a trade-off between $\sigma$ and the stabilization time.
In the following statement, $\T(t)$ denotes the \emph{real-time skew,} i.e., the maximum offset of output clocks relative to the external reference time $t$, and $D_H$ is the diameter of the virtual graph obtained by adding a virtual reference node connected to all nodes that have access to the external reference time.
For a simple statement, here we assume the respective error to be bounded by $\Delta$ and change by at most $\delta(T)$ in $T$ time as well.
\begin{corollary}\label{cor:external}
Suppose that $1+2(\vartheta-1)\le \zeta<1+\mu$, $T\ge C\Delta D/(\zeta-1)$ for a sufficiently large constant $C>0$, and $\sigma:=\mu/(\zeta-1)\ge 2$.
There is a GCS algorithm with stabilization time $\BO(C\Delta D/(\zeta-1))$ outputting clocks running at rates between $1/\zeta$ and $(\vartheta(1+\mu))/\zeta$ that achieves global and local skew bounds of
\begin{align*}
\T(t)\le \G(t)\le \left(1+\frac{3}{\sigma-1}\right)\Delta D_H \quad \mbox{and} \quad\Ll(t)\in 3\Delta+4\delta(T)(\log_{\sigma}D_H+\BO(1)).
\end{align*}
\end{corollary}
Here, $D_H\le D+1$ is the diameter of the graph obtained by adding a virtual node that is connected to each node with access to an external time reference.

\paragraph*{Syntonization.}
In order to leverage a frequency stability of local oscillators that is better than $\vartheta-1$, we \emph{syntonize} the local oscillators.
In contrast to synchronization, syntonization aims to match oscillator frequency to a reference, without considering clock offsets.
The standard approach to this are so-called \emph{phase-locked loops (PLLs),} which control a local oscillator based on measuring phase offset relative to a regular reference signal.
However, no attempt is made to control the absolute error of such a phase offset measurement;
the key requirement is that this error does not \emph{change} too suddenly.

This aligns with our approach in a natural way.
To achieve self-stabilization, we already needed to bound the global skew.
The same technique can be applied to measure the offset to the clock of a reference node---either inside the network, for internal synchronization, or the virtual node representing the external reference.
The error of up to $\W^s$ incurred by these measurements translates to a frequency error of the lock, which is bounded by $\BO(\W^s/P)$, with $P$ being the period of the reference signal used by the local PLLs.
One then simply uses the frequency-stabilized oscillators instead of free-running ones as the local references, allowing us to replace $\vartheta$ by $1+\BO(\nu(P)+\W^s/P)$, where $\nu(P)$ is the frequency stability of a locked PLL whose input perfectly reproduces the frequency of the reference.
\begin{theorem}\label{thm:frequency}
Increasing the stabilization time by the time a PLL with period $P\ge 2\Wd$ requires to lock,\footnote{I.e., the time required for the PLL to converge, achieving its guarantees under nominal operation.} we can replace the model parameter $\vartheta$ by $\vartheta'\in 1+\BO(\nu(P)+\W^s/P)$.
\end{theorem}

\subsection*{When Does it Matter?}

To illustrate where the presented improvements in the performance of GCS algorithms are of interest, let us briefly discuss three sample application scenarios.

\paragraph*{Internet Synchronization.}
Our first and negative example is synchronization via the Internet.
Here, NTP achieves worst-case errors in the tens to hundreds of milliseconds, while achieving ``better than one millisecond accuracy in local area networks under ideal conditions''~\cite{mills2010ntpv4}.
The main issue in the Internet are highly varying and asymmetric communication delays, rendering an approach based on stable clock offset measurements impractical.
While at first glance the factor 10-100 gap between ``global'' and ``local'' performance sounds promising, it is important to remember that we require an upper bound on how much measurement errors can change on suitable time scales.
Since there is no control over routing paths and queuing delays available to NTP in this setting, $\delta\approx \Delta$.
Moreover, the primary use case for NTP over the Internet is coarse synchronization to UTC, rendering the utility of local skew bounds questionable.

\paragraph*{Clock Distribution for Synchronous Hardware.}
As highlighted by the above example, relevant application scenarios satisfy two key criteria: (i) $\delta(T)\ll \Delta$ and (ii) the local skew being the key quality metric.
In~\cite{bund2020pals}, the authors propose to leverage GCS for clocking of large-scale synchronous hardware.
Here, the local skew in terms of physical proximity within the system is to be kept small:
as the communicating subcircuits are physically adjacent, the local skew limits the frequency at which the system can reliably be operated.
Combining a regular grid of GCS nodes with local clock distribution trees bears the promise of outperforming traditional clocking methods.
However, clearly the break-even point in terms of system size heavily depends on the performance of the GCS algorithm.

The dominant sources of clock offset error estimation between the GCS nodes of such a system are delay variations due to process variations, changes in temperature and voltage, and aging of the components.
These result in standard on-chip (ring) oscillators having impractically low accuracy.
Accordingly, such oscillators are locked to external quartz references to result in more accurate and stable frequencies.
This results in $\vartheta'-1$ around $10^{-6}$ for a (cheap) quartz oscillator that would be used as common reference for all nodes~\cite[Sec.~7.1.6]{horowitz2015art}.
However, clock speeds in modern systems are beyond one gigahertz, meaning that local skews of at most tens of picoseconds can be tolerated.
Thus, even for minimally chosen $\mu$, a back of the envelope calculation suggests $10^{-3}\,s$ for $D\le 100$ as a safe upper bound on $\W^s/\mu$.
This is below the time scales relevant for temperature~\cite{coskun2008temperature} and, certainly, aging.
As the effect of process variations is not time-dependent, voltage fluctuations are the remaining factor.
These need to be kept in check by modern systems to the best possible extent\footnote{Voltage droops due to sudden load changes cause significant timing variability and power integrity challenges, motivating adaptive voltage and frequency scaling techniques to maintain reliable operation; see, e.g.,~\cite{bowman13droop,fuegger22fast}. In contrast, the clocking subsystem draws steady power and can be protected from such droops by using dedicated supplies, which ensures stable operation under varying load conditions.} anyway, as they pose a key constraint on the achievable trade-off between clock speed and energy consumption, see e.g.~\cite{gupta2007voltage}.
Accordingly, the assumption that $\delta\ll \Delta$ is justified, and the presented results could improve the achievable local skew by an order of magnitude or more.

The proposed approach to self-stabilization might be too cumbersome in this context.
Instead, one should fall back to detecting locally whether the skew bounds the system was designed for are exceeded, in which case recovery can be triggered.
This very simple check is easily implemented and sufficient when it is enough to work with a ``hard'' limit on the local skew that is larger than~$\Delta$.

Note that accurate synchronization to UTC is usually not crucial in VLSI systems;
rather, the external oscillator serves as \emph{frequency} reference only.
However, the proposed synchronization approach could be seamlessly integrated within data centers, where bounding the local skew in the routing network would be beneficial for synchronous low-latency routing schemes, see e.g.~\cite{chatterjee2016globally}.
In this broader setting, accurate synchronization an external reference is of interest for generating globally valid timestamps~\cite{rakon2025blog}.

\paragraph*{Synchronization in Wireless and Cellular Networks.}
Further applications arise in wireless networks.
Especially for low-latency communication, the participants of such a network need to maintain tight synchronization.
This enables to align transmission time slots with small guardbands and avoids interference, which is of particular interest in cellular networks.
Moreover, performing time offset measurements via the wireless medium is equivalent to measuring distances.
This allows us to tightly synchronize base stations, which in turn enables passively listening devices to determine their location based on knowledge of the position of base stations and the relative arrival time of the signals of the base stations.
Given the limitations of GPS coverage both due to line-of-sight requirements and the prevalence of jamming, such an approach could complement satellite-based time and navigation services to improve availability and reliability.

Observe that also here, the local skew is decisive, not the global one:
communication happens with close-by devices, and interference falls off quickly at larger distances;
positioning is based on nearby basestations, as their signals are strongest.
While dynamic environments may cause fairly sudden delay changes, round-trip measurements are only susceptible to introduced asymmetry.
Mitigation techniques like taking the median of multiple measurements, dropping outliers, and temporarily removing unreliable links---it is known how to handle dynamic graphs~\cite{kuhn10dynamic,kuhn09}, and the techniques from this work are compatible with our approach---are likely to filter out the vast majority of cases where delays are unstable.
Noting that the GCS algorithm gracefully degrades in face of small, short-lived errors and can be made fully self-stabilizing, it is highly promising to exploit the typical case, i.e., relatively stable communication delays and hence offset measurements.

\subsection*{Technical Overview}
We now give an overview of the structure of the remainder of this paper and highlight the main technical challenges and ideas.
In \Cref{sec:model}, we introduce model and notation.
For the sake of a streamlined exposition, we do not justify the modeling choices in \Cref{sec:model}.
However, as this is crucial especially when proposing novel extensions, we offer an in-depth discussion to the interested reader in \Cref{app:discussion}.

\Cref{sec:algo} presents and analyses a baseline variant of our algorithm.
Interestingly, the algorithm itself is near-identical to prior GCS algorithms.
We follow~\cite{kuhn09reference} in our presentation, expressing the algorithm in the form of conditions it implements, which are ensured by adhering to the \emph{fast} and \emph{slow triggers} that can be expressed in terms of estimated clock offsets to neighbors.
The key difference is that the conditions and triggers scale with $\delta(T)$ instead of $\Delta$.

The main challenge lies in analyzing this variant of the algorithm.
In prior work, this was based on the potential functions
\begin{equation*}
\Psi^s(t):=\max_{v\in V}\{\Psi_v^s(t)\},\quad \mbox{where}\quad \Psi_v^s(t)=\max_{w\in V}\{L_w(t)-L_v(t)-d^s(v,w)\},
\end{equation*}
$s\in \N$, where $d^s(v,w)$ is the distance between $v$ and $w$ in the graph $(V,E,4s\Delta_e)$, i.e., edge $e$ has weight corresponding to its worst-case offset estimation error $\Delta_e$.
The conditions mentioned above ensure that (i) $\Psi_v^s(t)$ can grow only at rate $\vartheta-1$, as any maximizing node satisfies the slow trigger, and (ii) for $s\neq 0$, we are guaranteed that $L_v$ will make at least $\Psi_v^s(t)$ more than minimum progress by time $t+\Psi^{s-1}(t)/\mu$.
As $\Psi^0(t)=\G(t)$, induction over $s$ then shows that we can bound $\Psi^s(t+\BO(\G(t)/\mu))\le \G(t)/\sigma^s$, where $\sigma=\mu/(\vartheta-1)$.
Since for $\{v,w\}\in E$, we have that $d^s(v,w)\le 4s\Delta_e$, choosing $s=\lceil\G(t)/\Delta_e\rceil$ yields a bound of $\Ll_e(t)\le \BO(\Delta_e\log_{\sigma}\G/\mu)$ on the local skew on edge $e=\{v,w\}$ for times $t\ge t_0\in\BO(\G(0)/\mu)$, where $\G$ is an upper bound on the global skew at times $t\ge t_0'\in\BO(\G(0)/\mu)$ that is established by similar arguments.

The crucial insight is that this approach can be adapted to our setting by replacing the edge weights with a carefully chosen term.
The above potential considers $L_v(t)-L_w(t)=0$ as the ``target'' state for neighbors $v$, $w$, and by extension for all pairs $v,w\in V$.
Instead, fix a time $t_0$ and consider the offset $O_{v,w}$ required such that $L_v(t_0)-L_w(t_0)=O_{v,w}$ would result in $v$ \emph{measuring} a skew of $0$ to $w$;
note that this aligns with the intuition provided by our earlier example, in which an algorithm could not do better than matching the observed skews on each edge of the cycle, because it cannot know which edge(s) are responsible for its inability to reduce all observed skews to $0$.
Putting weight $4s\delta_{\{v,w\}}(T)-O_{v,w}$ on edge $(v,w)\in \vec{E}:=\bigcup_{\{v,w\}\in E}\{(v,w),(w,v)\}$, we can encode that we ``expect'' a skew of $O_{v,w}$ on the edge from the perspective of the analysis.
The corresponding adjustment to the conditions and triggers shifts the thresholds causing clocks to run fast or slow by $O_{v,w}$ as well, meaning that the offset measurements nodes take are now ``accurate'' from the perspective of the modified potential up to a remaining error of $\delta(T)$ at times $t\in [t_0-T,t_0+T]$.
Bounding the potential similarly to above, one would then hope to prove that  $\Ll_{\{v,w\}}(t)\le \max\{O_{v,w},O_{w,v}\}+\BO(\G(0)/\mu)$.

Naturally, there is a catch.
Since edge weights of the graph $(V,\vec{E},4s\delta_{\{v,w\}}(T)-O_{v,w})$ can be negative, it is possible that a negative cycle is created, and there are no well-defined distances.
This corresponds to a situation in which the potentials $\Psi_v^s(t)$ are not meaningful;
in particular, we cannot guarantee the key properties enabling the above induction.
Fortunately, once $s$ is large enough, negative cycles are prevented, and the desired properties are restored.
Our earlier example illustrates why this is a necessity rather than an artifact of our analysis.

Note that this problem arises even in acyclic graphs, as we replaced each undirected edge by two directed ones.
Addressing this issue, we require that $O_{v,w}=-O_{w,v}$, which is ensured by having both $v$ and $w$ derive their offset estimates from the same measurements.
As clocks run at slightly different speeds, this requires to account for additional errors, but these can be subsumed by $\delta_{\{v,w\}}(T)$.\footnote{As neighbors need to communicate in order to determine the estimates anyway, sending the computed value to the respective other node has small impact on $\delta_{\{v,w\}}$.}
This eliminates the possibility of cycles of length $2$, strengthening the obtained bounds.
Moreover, it means that graphs of high girth are likely to have well-defined potentials $\Psi_v^s$ for small values of $s$;
in particular, $s=0$ is valid on trees, meaning that the algorithm's performance is provably close to that of tree-based protocols like NTP or PTP when using the same tree for synchronization.

With the modifications to the model, conditions, triggers, and potentials in place, the remaining technical contribution of \Cref{sec:algo} is to adapt the analysis to these new notions.

In \Cref{sec:stab}, we move on to addressing self-stabilization.
The algorithm presented in \Cref{sec:algo} has the advantage of acting solely based on local information, i.e., the offset estimates to neighbors, the values of $\delta_e(T)$ for incident edges $e$, and the progress of time measured by the local oscillator.
However, given the possibility of negative cycles considered above, it is impossible to decide whether skew bounds are violated based on this knowledge only, unless one falls back to (highly) conservative bounds.

The analysis from \Cref{sec:algo} requires the constraint that $T\ge C\W^s/\mu$ for a sufficiently large constant $C>0$, where $s\in \N$ is minimal such that $\Psi^{s-1/2}(t)$ is well-defined, i.e., we have no negative cycles in the corresponding graph.
In \Cref{sec:stab}, we show how to estimate $\Psi^{s'}$ for some $s'\in [s,s+\BO(1)]$ up to an error of $\W^{s'}$, within time proportional to the maximum communication delay for (multi-hop) communication between any pair of nodes;
note that the latter is a trivial lower bound on the stabilization time.
This enables us to execute the standard approach of repeatedly approximating $\Psi^{s'}(t)$ and resetting the logical clocks if this skew exceeds the bounds that hold in absence of transient faults.
The analysis from \Cref{sec:algo} then shows that the bounds are established within $\BO(\W^{s'}/\mu)$ time once $\Psi^{s'}$ is small enough.
In other words, the algorithm can be made self-stabilizing with the minor caveat that we have to replace $s$ by $s'$.

This replacement is a result of a key technical challenge arising from the need to obtain non-local knowledge.
If we could take an instantaneous ``snapshot'' of the perceived offsets between any pair of neighbors, it would be possible to determine $s$ exactly.
However, in order to collect and utilize this information, non-local communication is required.
The time this takes causes additional errors, which means that learning $s$ (which actually even depends on the chosen instance of time) exactly is, in general, infeasible.

We collect the necessary information via a shortest-path tree w.r.t.\ the weights given by the bounds on communication delay. 
The main technical challenge overcome in \Cref{sec:stab} is showing how to use the collected values to estimate $\Psi^{s'}$.

Next, in \Cref{sec:external}, we discuss our approach to external synchronization.
The main idea is to add a virtual node $v_0$ whose clock is equal to the real time.
Any node that has access to an external reference can use this to simulate an edge to this reference, with the quality of the reference determining the error on the edge.
The advantage of this method is that we can use our machinery with almost no modification.
However, some challenges arise.
\begin{compactitem}
  \item The virtual node's clock always equals the real time. Yet, we formally want it to satisfy the conditions of the GCS algorithm, so that our analysis applies. We achieve this by slowing down the clocks of all actual nodes by a small factor $\zeta>1$, so that the virtual node never comes into the situation that it satisfies the fast condition.
  \item For the latter, it is sufficient that $\Psi^s_{v_0}=0$. As we slowed down the clocks of nodes $v\neq v_0$, $\Psi^s_{v_0}=0$ becomes an invariant as soon as it holds for some time $t$. Prior to this time, $\Psi^s_{v_0}$ decreases at rate $\zeta-\vartheta$. Thus, large $\zeta$ facilitates getting to this point quickly.
  \item On the other hand, the fact that $v_0$ may run by factor $\zeta>\vartheta$ faster than the slowest clock in the system means that $\Psi^s_v$ can increase at rate $\zeta-1>\vartheta-1$ for $v\neq v_0$. Hence, we want to keep $\zeta-1$ small compared to $\mu$, since in the skew bounds now $\sigma=\mu/(\zeta-1)$.
  \item When bounding $\Psi^{s'}$ for the purpose of self-stabilization, we need to do so for the virtual graph, i.e., take into account $v_0$. The resulting technical difficulties can be resolved by revisiting the respective computations and exploiting that $L_{v_0}(t)=t$ at all times $t$.
\end{compactitem}

In \Cref{sec:syntonization}, we move on to discussing how to rein in the frequencies of the local oscillators using some reference, either at a node in the system or an external one.
Conceptually, this is straightforward: we regularly generate non-local offset estimates between the reference and each node.
Technically, the respective heavy lifting has already been done in \Cref{sec:stab}, as a necessary ingredient to estimating $\Psi^{s'}$.


\section{Model and Notation}\label{sec:model}

We model the distributed system as an undirected connected graph $G=(V,E)$.
We assume that there is a designated unique root\footnote{This is essentially equivalent to assuming unique IDs, as the root can assign unique IDs in a self-stabilization way. Note that a significant share of our results, in particular those in \Cref{sec:algo} including \Cref{cor:uniform}, are valid in any network in which nodes can distinguish between their neighbors.} node $r\in V$.

\paragraph*{Hardware clocks.}
Each node $v\in V$ is equipped with a \emph{hardware clock} of bounded frequency error $\vartheta-1$.
Formally, we model this clock as a differentiable function $H_v:\Rp\to \Rp$ satisfying that $1\le \frac{dH_v}{dt}(t)\le \vartheta$, i.e., hardware clock rates vary between $1$ and $\vartheta$.
Nodes have query access to their hardware clocks, i.e., can read them at any time, but cannot access the underlying ``true'' time $t$.
Note that there are no guaranteed bounds on $H_v(t)-t$ at any time $t$, so without access to an external reference, the system cannot synchronize to the ``true'' time.
We also consider the possibility that these clocks are more stable over bounded time periods.
However, for the sake of presentational flow, we defer the respective modeling and discussion to \Cref{sec:syntonization}.

\paragraph*{External reference clocks.}
A (possibly empty) subset $R\subseteq V$ of the nodes has access to an external time reference, e.g.\ via a GPS receiver providing it with a time closely aligned to UTC.
As neither the external reference nor the means by which such a node $v\in R$ learns its time can be perfectly accurate, we model this as follows.
Each $v\in R$ can query its reference at any time $t\in \Rp$, yielding an estimate $\tilde{t}\in \Rp$ of \emph{unknown} error.
However, we assume that this error changes relatively slowly over time.
That is, for $T\in \Rp$ there is a known value $\delta_v(T)$ with the property that for all $t_0\in \Rp$ and $t,t'\in [t_0,t_0+T]$, it holds that
\begin{align*}
|\tilde{t}-t-(\tilde{t}'-t')|\le \delta_v(T).
\end{align*}
The task for an algorithm is then to track real time as closely as the unknown error permits, while having knowledge of $\delta_v(T)$ only.
\paragraph*{Logical clocks and output requirements.}
The goal of a synchronization algorithm is to output at each node a logical clock $L_v$ that (i) always advances at a rate close to $1$, (ii) always has a small offset to other nodes' logical clocks, and (iii) if some node has access to an external reference, is also close to $t$. 
However, we aim for a self-stabilizing solution.
Therefore, we require these properties only at times $t\ge S$, where $S$ is the \emph{stabilization time} of the algorithm.
Note that minimizing $S$ is an optimization goal in its own right.

Formally, node $v\in V$ maintains logical clock values $L_v(t)$ such that $L_v:\Rp\to \Rp$ satisfies for all $t,t'\in \Rp$, $t'\ge t$, that
\begin{equation*}
\alpha(t'-t)\le L_v(t')-L_v(t) \le \beta(t'-t),
\end{equation*}
where we seek to minimize both the maximum rate difference $\beta-\alpha$ and the deviation $|1-(\alpha+\beta)/2|$ from a nominal rate of $1$.\footnote{Without external reference, one may normalize rates by scaling all logical clocks by factor $1/\alpha$, enforcing a minimum rate of $1$. Then both criteria are equivalent to minimizing $\beta$ subject to the constraint that $\alpha=1$.}

In terms of quality of synchronization, we seek to minimize the \emph{global skew}
\begin{equation*}
\G(t):=\max_{v,w\in V}\{L_v(t)-L_w(t)\},
\end{equation*}
the \emph{local skew,} i.e.,
\begin{equation*}
\Ll(t):=\max_{e\in E}\{\Ll(e)\}, \quad \mbox{where}\quad \Ll_{\{v,w\}}(t):=|L_v(t)-L_w(t)|,
\end{equation*}
and the \emph{real-time skew}\begin{equation*}
\T(t):=\max_{v\in V}\{|L_v(t)-t|\}.
\end{equation*}
Note that always $\Ll(t) \le \G(t) \le 2\T(t)$, but the gaps can become arbitrarily large.
Moreover, no bound on $\T$ can be guaranteed in the case of internal synchronization, i.e., when $R=\emptyset$.
\paragraph*{Clock offset estimates.}
Node $v\in V$ maintains for each neighbor $w$ an estimate $o_{v,w}(t)$ of $L_v(t)-L_w(t)$, whose error $e_{v,w}(t):=L_v(t)-L_w(t)-o_{v,w}(t)$ does not satisfy a known bound.
For convenience, we assume that $o_{v,w}(t)$ changes only at discrete points in time, i.e., $o_{v,w}(t)$ is piece-wise constant.\footnote{Since the logical clocks $L_v$ increase at bounded rate, this can be achieved by sampling any $o_{v,w}(t)$ without this constraint at least once every $\varepsilon\in \R_{>0}$ time and increasing $\Delta_{v,w}$ by $(\beta-\alpha)\varepsilon$.} We require that this error is approximately antisymmetric, i.e., $e_{v,w}(t)\approx -e_{w,v}(t)$, and that it changes relatively slowly over time.
For convenient use in our analysis, we merge both requirements into a single constraint.
That is, for $T\in \Rp$ there is a known value $\delta_{\{v,w\}}(T)$ such that for all $t_0\in \Rp$ and $t,t'\in [t_0,t_0+T]$, it holds that
\begin{align*}
|e_{v,w}(t')-e_{v,w}(t)|&< \delta_{\{v,w\}}(T)\quad
\mbox{and}\quad |e_{v,w}(t')+e_{w,v}(t)|< \delta_{\{v,w\}}(T).
\end{align*}
In essence, any estimation with slowly changing error can be turned into one with the required approximate antisymmetry by taking one node's estimate and communicating it to the other, which merely flips the sign.

\paragraph*{Explicit communication.}
While clock offset estimates are sufficient for the ``core'' of the synchronization algorithm, non-local communication is required for some subroutines.
Accordingly, we assume that for $e=\{v,w\}\in V$, nodes $v$ and $w$ can exchange messages with a known upper bound on the maximum end-to-end \emph{delay} $d_e$.
That is, if $v$ decides to send a message to $w$ according to a computational step at time $t$, then $w$ receives the message and executes a computational step processing it at a time from $(t,t+d_e]$.

It would be possible to fully parametrize all of our results by explicitly using the values $d_e$.
However, for the sake of simpler statements, we assume that there is a constant $c>0$ so that $c\delta_e\ge (\beta-\alpha) d_e$ for each $e\in E$.
Intuitively, this condition states that clock offset estimates are not better than the explicit communication mechanism permits:
if it takes $d_{\{v,w\}}$ time to send a message from $v$ to $w$ or vice versa, in the meantime $|L_v-L_w|$ may change by $(\beta-\alpha) d_{\{v,w\}}$.

\paragraph*{Computational model.}
To avoid the notion of continuous computations, the algorithm is only required to output $L_v$ when queried at discrete points in time, which may be chosen adversarially.
The computation of $L_v(t)$ when queried at time $t$ has no further effect on the execution, i.e., a query only ``reads'' the logical clock.
Note that since $v$ \emph{could} be queried at any time and $L_v$ is continuous, this implies that an algorithm specifies $L_v(t)$ at all times $t\in \Rp$, consistent with the above definitions.

Each node executes a state machine, where several events may trigger a computational step of $v\in V$ at time $t$:
\begin{inparaenum}[(i)]
  \item $o_{v,w}$ changes for some $\{v,w\}\in E$;
  \item $v$ receives a message from a neighbor; or
  \item $H_v(t)=h$ for a previously computed (or hard-coded) hardware clock value.
\end{inparaenum}
Algorithms are not allowed to (ab)use (iii) to cause an unbounded number of computational steps in finite time.
Moreover, we assume that there is a positive minimum communication delay.
Since $o_{v,w}$ changes only at discrete points in time, this means that the number of computational steps a state machine takes in a finite time interval is finitely bounded.

In a computational step taking place at time $t$, the state machine receives as input $H_v(t)$ and $o_{v,w}(t)$ for each $\{v,w\}\in E$, as well as any messages received, where $v$ can distinguish between neighbors (e.g.\ by a port numbering).
It then updates its state and possibly computes messages to send to some of its neighbors;
an algorithm is specified by (i) its state variables, (ii) determining how these computations are performed in line with the stated constraints, and (iii) specifying how $L_v(t)$ is computed when queried at time $t$, based on $H_v(t)$ and $v$'s state.

Recall that we desire a self-stabilizing solution, which is typically modelled by memory (including network buffers) having an arbitrary initial state, and the system needing to converge to some desirable behavior within bounded time.
Analogously, we assume that the memory of the state machine of $v\in V$ is in an arbitrary state at time $0$.
However, for convenience we assume that no spurious messages are in transit anymore, i.e., any message received at time $0$ or later has been computed in accordance with the nodes' states and the specified algorithm.
As end-to-end communication delays are bounded, this assumption becomes valid no later than $\max_{e\in E}\{d_e\}$ time after transient faults cease.

\paragraph*{Notation.}
We denote by $\N$ the non-negative integers.
For $k\in \N$, we denote $[k]_0:=[0,k]\cap \N$ and $[k]:=[k]_0\setminus \{0\}$.
When we write $x=y\pm z$, this is equivalent to $|x-y|\le z$.
A \emph{walk} is any sequence of nodes $W=v_0,\ldots,v_{\ell}$ such that for each $i\in [\ell]$ it holds that $\{v_{i-1},v_i\}\in E$.
If the graph is weighted with edge weights $\omega$, then the walk has weight $\omega(W):=\sum_{i=1}^{\ell}\omega(v_{i-1},v_i)$.
The set of walks from $v\in V$ to $w\in V$ is denoted by $W_{v,w}$.
Analogous definitions are used for directed graphs.
If for a directed edge $e=(v,w)$ a value $\mathrm{val}_e$ is not defined, but $\mathrm{val}_{\{v,w\}}$ is, $\mathrm{val}_e$ is interpreted as $\mathrm{val}_{\{v,w\}}$.

Throughout the paper, unless stated otherwise, we assume that all considered times lie within an interval $[t,t+T]$.
We will prove a bound that $T$ must satisfy in order to ensure that the stabilization time $S\le T/2$, implying that we can cover $\Rp$ by the overlapping intervals $[iT/2+S,iT/2+T]$, $i\in \N$, so that for each interval we can apply the derived bounds on skews etc.\ during $[iT/2+S,iT/2+T]\subseteq [(i+1)T/2,iT/2+T]$.
This covers the interval $[S,\infty)$, i.e., establishes all bounds at times $t\ge S$, as required for an algorithm with stabilization time $S$.

\section{Basic GCS Algorithm}\label{sec:algo}

In the following, we fix $T\in \Rp$ and will specify the algorithm $\A(T)$.
To do so, we follow the approach of first specifying the key properties the algorithm needs to satisfy to achieve strong bounds on the local skew $\Ll$.
Once specified, implementing these properties is extremely straightforward.

\subsection*{Conditions and Triggers}

We start by expressing the desired conditions in terms of the logical clock offsets to neighbors, which nodes do not have access to.
In prior work, it was assumed that there is a known bound $\Delta_{\{v,w\}}$ such that $|o_{v,w}(t)|\le \Delta_{\{v,w\}}$ for \emph{all} $t\in \Rp$.
The imposed conditions then would require a node to advance its logical clock at rate $\frac{dH_v}{dt}(t)$ if for some $s\in\N$ (i) there is some neighbor whose logical clock lags by $2s\Delta_{\{v,w\}}$ behind and (ii) no neighbor's clock is by $2s\Delta_{\{v,w\}}$ ahead of $L_v(t)$.
Similarly, the logical clock rate should be $(1+\mu)\frac{dH_v}{dt}(t)$ for a choosable parameter $\mu>\vartheta-1$ if (i) there is some neighbor whose logical clock is by more than $(2s+1)\Delta_{\{v,w\}}$ ahead and (ii) no neighbor's clock is by more than $(2s+1)\Delta_{\{v,w\}}$ behind $L_v(t)$.
The ``separation'' of these rules by an offset of $\Delta_{\{v,w\}}$ makes it possible to simultaneously implement both of them despite nodes having access only to inaccurate estimates $o_{\{v,w\}}(t)$ of $L_v(t)-L_w(t)$ for each neighbor $w$.

We seek to replace $\Delta_{\{v,w\}}$ by the much smaller $\delta_{\{v,w\}}:=\delta_{\{v,w\}}(T)$, but doing this naively entails that the estimates cannot be used to implement the conditions.
Instead, we must accept that the conditions that are realized incur a large bias of up to $\Delta_{\{v,w\}}$ (where $\Delta_{\{v,w\}}$ is not known).
However, as we analyze a time interval of length $T$, we have the guarantee that the error of the estimates \emph{changes} little.

With this in mind, for an arbitrary $a\in \Rp$, let us fix a time interval $[a,a+T]$ throughout this section.
All times will be from this interval only, so that for any considered times $t$ and $t'$ and $\{v,w\}\in E$, it holds that $|e_{v,w}(t')-e_{v,w}(t)|< \delta_{\{v,w\}}(T)$ and $|e_{v,w}(t')+e_{w,v}(t)|< \delta_{\{v,w\}}(T)$.
For this interval, for each $\{v,w\}\in E$ we can fix a ``nominal'' offset around which the estimation error $e_{\{v,w\}}(t)$ varies.
\begin{definition}[Nominal Offset]
For $\{v,w\}\in E$, set $O_{v,w}:=(e_{v,w}(a+T/2)-e_{w,v}(a+T/2))/2$.
\end{definition}
Note that this definition ensures not only that $|e_{v,w}(t)-O_{v,w}|<\delta_{\{v,w\}}$ for all $t\in [a,a+T]$, but also that $O_{v,w}=-O_{w,v}$.
Thus, if $L_v(t)-L_w(t)-O_{v,w}=0$, then also $L_w(t)-L_v(t)-O_{w,v}=-(L_v(t)-L_w(t)-O_{v,w})=0$.
In this sense, $O_{v,w}$ is going to be the ``nominal'' offset:
our conditions will be such that $L_v(t)-L_w(t)=O_{v,w}$ is considered the ``desired'' state of edge $\{v,w\}$.
This leads to the following variation of the slow and fast conditions introduced by Kuhn and Oshman~\cite{kuhn09reference}.
\begin{definition}[Slow Condition]\label{def:sc}
The \emph{slow condition} holds at $v\in V$ at time $t\in \Rp$ if and only if there is some \emph{level} $s\in \N$ such that
\begin{align*}
\exists \{v,w\}\in E\colon L_v(t)-L_w(t)-O_{v,w} &\ge 4s\delta_{\{v,w\}}\\
\forall \{v,w\}\in E\colon L_v(t)-L_w(t)-O_{v,w} &\ge -4s\delta_{\{v,w\}}.
\end{align*}
\end{definition}

\begin{definition}[Fast Condition]\label{def:fc}
The \emph{fast condition} holds at $v\in V$ at time $t\in \Rp$ if and only if there is some \emph{level} $s\in \N$ such that
\begin{align*}
\exists \{v,w\}\in E\colon L_v(t)-L_w(t)-O_{v,w} &\le -(4s+2)\delta_{\{v,w\}}\\
\forall \{v,w\}\in E\colon L_v(t)-L_w(t)-O_{v,w} &\le (4s+2)\delta_{\{v,w\}}.
\end{align*}
\end{definition}
Intuitively, the difference of the above conditions to~\cite{kuhn09reference} is that we re-adjust the goal from seeking to minimize the skew on edge $\{v,w\}\in E$ towards minimizing the \emph{apparent} skew when considering $O_{v,w}$ the baseline.
From the point of view of the measurement on this particular edge, this baseline is the best we can hope for without knowledge of $O_{v,w}$, as with $L_v(t)-L_w(t)=O_{v,w}$ we might have $o_{v,w}(t)=o_{w,v}(t)=0$, i.e., neither $v$ nor $w$ observe any skew on the edge.

Next, we translate these conditions into ``triggers'' based on the nodes' estimates of clock offsets.
The triggers are chosen such that if the slow resp.\ fast condition holds, then the corresponding trigger holds, too.
\begin{definition}[Slow Trigger]\label{def:st}
The \emph{slow trigger} holds at $v\in V$ at time $t\in \Rp$ if and only if there is some \emph{level} $s\in \N$ such that
\begin{align*}
\exists \{v,w\}\in E\colon o_{v,w}(t) &> (4s-1)\delta_{\{v,w\}}\\
\forall \{v,w\}\in E\colon o_{v,w}(t) &> -(4s+1)\delta_{\{v,w\}}.
\end{align*}
\end{definition}

\begin{definition}[Fast Trigger]\label{def:ft}
The \emph{fast trigger} holds at $v\in V$ at time $t\in \Rp$ if and only if there is some \emph{level} $s\in \N$ such that
\begin{align*}
\exists \{v,w\}\in E\colon o_{v,w}(t) &< -(4s+1)\delta_{\{v,w\}}\\
\forall \{v,w\}\in E\colon o_{v,w}(t) &< (4s+3)\delta_{\{v,w\}}.
\end{align*}
\end{definition}

\begin{lemma}\label{lem:implies}
If the slow (fast) condition holds at $v\in V$ at time $t\in \Rp$, then the slow (fast) trigger holds at $v\in V$ at time $t\in \Rp$.
\end{lemma}
\begin{proof}
Let $s$ be a level for which the slow condition holds at $v\in V$ at time $t$.
Then there is some $\{v,w\}\in E$ such that
\begin{align*}
o_{v,w}(t)&= L_v(t)-L_w(t)-e_{v,w}(t)\\
&\ge L_v(t)-L_w(t)-O_{v,w}-|O_{v,w}-e_{v,w}(t)|\\
&>L_v(t)-L_w(t)-O_{v,w}-\delta_{\{v,w\}}\\
&\ge (4s-1)\delta_{\{v,w\}}.
\end{align*}
Similarly, for all $\{v,w\}\in E$, the condition implies that
\begin{equation*}
o_{v,w}(t)> L_v(t)-L_w(t)-O_{v,w}-\delta_{\{v,w\}} \ge -(4s+1)\delta_{\{v,w\}},
\end{equation*}
i.e., the slow trigger holds at $v$ at time $t$.

Now let $s$ be a level for which the fast condition holds at $v\in V$ at time $t$.
Then there is some $\{v,w\}\in E$ such that
\begin{equation*}
o_{v,w}(t)< L_v(t)-L_w(t)-O_{v,w}+\delta_{\{v,w\}}< -(4s+1)\delta_{\{v,w\}}.
\end{equation*}
Similarly, for all $\{v,w\}\in E$, the condition implies that
\begin{equation*}
o_{v,w}(t)< L_v(t)-L_w(t)-O_{v,w}+\delta_{\{v,w\}}< (4s+3)\delta_{\{v,w\}},
\end{equation*}
i.e., the fast trigger holds at $v$ at time $t$.
\end{proof}

\begin{lemma}\label{lem:mutex}
The slow and fast trigger are mutually exclusive.
\end{lemma}
\begin{proof}
Let $s$ be a level for which the slow trigger holds at $v\in V$ at time $t$ and consider level $s'\in \N$.
If $s'\ge s$, observe that the slow trigger entails for all $\{v,w\}\in E$ that
\begin{equation*}
o_{v,w}(t)> -(4s+1)\delta_{\{v,w\}}\ge -(4s'+1)\delta_{\{v,w\}},
\end{equation*}
so the first requirement of the fast trigger cannot be satisfied on level $s'$.
If $s'<s$, the slow trigger implies that there is some $\{v,w\}\in E$ such that
\begin{equation*}
o_{v,w}(t)> (4s-1)\delta_{\{v,w\}}\ge (4s'+3)\delta_{\{v,w\}},
\end{equation*}
so the second requirement of the fast trigger cannot be satisfied on level $s'$.
\end{proof}

\subsection*{Algorithm}

As stated above, we wish for our algorithm to run ``fast,'' i.e., increase $L_v$ at rate $(1+\mu)\frac{dH_v}{dt}$ whenever the fast trigger holds, and run ``slow,'' i.e., increase $L_v$ at rate $\frac{dH_v}{dt}$ whenever the slow trigger holds.
If neither applies, any rate between the two extremes is a valid choice.
As the estimates $o_{v,w}$ are piece-wise constant and any change of an estimate causes $v$ to execute a computational step of its state machine, this is straightforward to realize, see \Cref{alg:gcs_basic}.
\begin{algorithm2e}[t]
	\small
	\DontPrintSemicolon
	\caption{Basic GCS algorithm, code at $v\in V$. It will converge to small skews for any initial values of the variables. However, the time to converge scales with $\G(0)/\mu$.\label{alg:gcs_basic}}
	\textbf{variables:} $L_v\in \Rp$, $H_v\in \Rp$, and $r_v\in \{1,1+\mu\}$\medskip\\
	
	When executing a computational step at time $t$:\\
	$L_v\gets L_v+r_v\cdot (H_v(t)-H_v)$\\
	$H_v\gets H_v(t)$\\
	\lIf{fast trigger holds at $v$}{
		$r_v\gets 1+\mu$
	}
	\lElse{
		$r_v\gets 1$
	}\medskip
	
	When $H_v(t)\bmod  \N = 0$:\\
	execute a computational step\medskip
	
	When queried for $L_v$ at time $t$:\\
	\KwRet{$L_v+r_v\cdot (H_v(t)-H_v)$}
\end{algorithm2e}
Clearly, \Cref{alg:gcs_basic} specifies an algorithm in the sense of \Cref{sec:model}.
For simplicity, the algorithm checks whether $r_v$ is consistent with the offset estimates once every time unit (according to $H_v$).
However, any suitable frequency for such consistency checks is valid.
In the following, we assume that $a\ge 1$, i.e., each node checked for state consistency at least once and $r_v=1+\mu$ if and only if the fast trigger holds at $v$.
\begin{lemma}\label{lem:implements}
Defining $r_v(t):=1+\mu$ if the fast trigger holds at $v\in V$ at time $t$ and $r_v(t):=1$ otherwise, with \Cref{alg:gcs_basic} we have that $L_v(t)=L_v(1)+\int_1^t r_v(\tau)\frac{dH_v}{dt}(\tau)\,d\tau$.
In particular, for any time $t$ at which $r_v$ does not change, \Cref{alg:gcs_basic} ensures that
\begin{equation*}
\frac{dL_v}{dt}(t)=r_v(t)\cdot\frac{dH_v}{dt}(t)=\begin{cases}
(1+\mu)\cdot\frac{dH_v}{dt}(t) &\mbox{if the fast trigger holds at $v$ at time $t$}\\
\frac{dH_v}{dt}(t) &\mbox{else.}
\end{cases}
\end{equation*}
\end{lemma}
\begin{proof}
Whether the fast trigger holds at $v\in V$ can only change when $o_{v,w}$ changes.
Since $o_{v,w}$ is piece-wise constant and any change results in execution of the code given in \Cref{alg:gcs_basic}, after the first computational step of $L_v$
the variable $r_v$ equals $1+\mu$ if and only if the fast trigger holds, and $r_v=1$ else. 
The first such step is taken at a time $t_v\le 1$.
Hence, by induction on the computational steps of $v$ we see that $L_v(t)=L_v(t_v)+\int_{t_v}^t r_v(\tau)\frac{dH}{dt}(\tau)\,d\tau$, where we use that $\int_{t_0}^{t_1}r\cdot\frac{dH_v}{dt}(\tau)\,d\tau=r\cdot(H_v(t_1)-H_v(t_0))$ for any $t_1>t_0\ge t_v$ and constant $r$.
In particular, for $t\ge 1\ge t_v$, we get that indeed $L_v(t)=L_v(1)+\int_{1}^t r_v(\tau)\frac{dH}{dt}(\tau)\,d\tau$.
As taking the derivative of an integral after its upper integration boundary equals the integrand evaluated at this boundary, the remaining claim follows.
\end{proof}
Rate bounds on $L_v$ are immediate from the above relation between $L_v$ and $H_v$.
\begin{corollary}\label{cor:rates}
\Cref{alg:gcs_basic} guarantees $\alpha=1$ and $\beta=(1+\mu)\vartheta$.
\end{corollary}

\subsection*{Potentials}

At the heart of analyzing \Cref{alg:gcs_basic} lie the following potential functions.
\begin{definition}\label{def:potential}
For $s\in \N$, we define the \emph{level-$s$ potential at $v\in V$} as
\begin{equation*}
\Psi_v^s(t):=\sup_{w\in V}\left\{\sup_{v_0,\ldots,v_{\ell} \in W_{v,w}}\left\{L_w(t)-L_v(t)-\sum_{i=1}^{\ell}(4s\delta_{\{v_{i-1},v_i\}}-O_{v_{i-1},v_i})\right\}\right\}
\end{equation*}
and the \emph{level-$s$ potential} as $\Psi^s(t):=\sup_{v\in V}\{\Psi_v^s(t)\}$.
\end{definition}
The key difference between this potential and prior work is the inclusion of the term $O_{v_{i-1},v_i}$, which is necessary to account for the inclusion of the same term in the slow and fast conditions.
However, this calls into question whether this potential can yield useful results.
In fact, in general it is not even bounded, as there might be a cycle in which $4s\delta_{\{v_{i-1},v_i\}}-O_{v_{i-1},v_i}$ sums to a negative value.
\begin{definition}\label{def:graph}
For $s\in 1/2\cdot\N$, define the \emph{level-$s$} graph as the weighted directed graph $G^s:=(V,\vec{E},\omega^s)$, where $\vec{E}:=\bigcup_{\{v,w\}\in E}\{(v,w),(w,v)\}$ and $\omega^s(v,w):=4s\delta_{\{v,w\}}-O_{v,w}$ for all $(v,w)\in \vec{E}$.
Let $s_0\in \N$ be minimal such that for $s\ge s_0+1/2$, $G^s$ contains no negative (directed) cycle.
For $s\ge s_0+1/2$, we denote by $d^s(v,w)$ the distance from $v$ to $w$ in $G^s$ and by $\W^s:=\max_{v,w\in V}\{d^s(v,w)\}$ the diameter of $G^s$.
\end{definition}
We will reason about levels $s>s_0$.
Before doing so, let us first point out that $s_0$ is bounded.
We use the following shorthand for this purpose.
\begin{definition}[Maximum Absolute Estimate Error]
For each edge $\{v,w\}\in E$, define $\Delta_{\{v,w\}}:=\max_{t\in [a,a+T]}\{|e_{v,w}(t)|,|e_{w,v}(t)|\}$.
\end{definition}
\begin{lemma}\label{lem:s0}
It holds that $s_0\le \max_{\{v,w\}\in E}\{\lceil\Delta_{\{v,w\}}/(4\delta_{\{v,w\}})\rceil-1/2\}$.
\end{lemma}
\begin{proof}
Let $s\ge \max_{\{v,w\}\in E}\{\lceil\Delta_{\{v,w\}}/(4\delta_{\{v,w\}})\rceil\}$.
Then for each $(v,w)\in \vec{E}$, we have that $\omega^s(v,w)\ge \Delta_{\{v,w\}}-O_{v,w}\ge \Delta_{\{v,w\}}-(|e_{v,w}(a+T/2)|+|e_{w,v}(a+T/2)|)/2\ge 0$, so $G^s$ cannot contain a negative cycle.
\end{proof}
For $s>s_0$, we have a well-defined distance function, allowing to simplify how $\Psi_v^s$ is expressed. 
\begin{lemma}\label{lem:valid}
For all $s_0< s\in \N$ and $v\in V$, it holds that
\begin{equation*}
\Psi_v^s(t)=\max_{w\in V}\left\{L_w(t)-L_v(t)-d^s(v,w)\right\}.
\end{equation*}
\end{lemma}
\begin{proof}
Consider any walk $W=v_0,\ldots,v_{\ell}\in W_{v,w}$ of weight $\omega^s(W)=\sum_{i=1}^{\ell}\omega^s(v_{i-1},v_i)$.
As $G^s$ contains no negative cycles, we know that this sum is bounded from below, and that the minimum is attained by a shortest path from $v$ to $w$.
The claim follows, as by definition
\begin{equation*}
\Psi_v^s(t)=\sup_{w\in V}\left\{\sup_{W \in W_{v,w}}\left\{L_w(t)-L_v(t)-\omega^s(W)\right\}\right\}.\qedhere
\end{equation*}
\end{proof}
As distances in $G^s$ can be negative even for $s>s_0$, $\Psi^s(t)$ can exceed $\G(t)$.
However, its deviation from $\G(t)$ is bounded by $\W^s$, which follows from the following helper statement.
\begin{lemma}\label{lem:reversal}
For all $v,w\in V$ and levels $s_0<s\in \N$, it holds that $-d^s(v,w)\le d^s(w,v)$.
\end{lemma}
\begin{proof}
Let $P\in W_{v,w}$ and $Q\in W_{w,v}$ be shortest paths w.r.t.\ edge weights $\omega^s_e$.
As the concatenation of $Q$ and $P$ forms a cycle, the definition of $s_0$ implies that
\begin{equation*}
d^s(v,w)+d^s(w,v)=\omega^s(P)+\omega^s(Q)\ge 0.\qedhere
\end{equation*}
\end{proof}

Our goal is to bound the potentials $\Psi^s(t)$ for $s_0<s\in \N$ and suitable times $t$, as they readily imply bounds on $\G(t)$ and $\Ll(t)$.
\begin{corollary}\label{cor:global_bound}
For each $s_0<s\in 1/2\cdot \N$, we have that $|\G(t)-\Psi^s(t)|\le\W^s$.
\end{corollary}
\begin{proof}
Let $v,w\in V$ such that $\G(t)=L_w(t)-L_v(t)$.
By \Cref{lem:valid},
\begin{equation*}
\G(t)=L_w(t)-L_v(t)\le \Psi_v^s(t)+d^s(v,w)\le \Psi^s(t)+\W^s.
\end{equation*}
Now consider $v,w\in V$ satisfying $\Psi^s(t)=L_w(t)-L_v(t)-d^s(v,w)$.
By \Cref{lem:reversal},
\begin{equation*}
\Psi^s(t)=L_w(t)-L_v(t)-d^s(v,w)\le L_w(t)-L_v(t)+d^s(w,v)\le \G(t)+\W^s.\qedhere
\end{equation*}
\end{proof}
\begin{corollary}\label{cor:local_bound}
For $\{v,w\}\in E$, it holds that
$\Ll_{\{v,w\}}(t) \le \min_{s_0<s\in \N}\{\Psi^s(t)+4s\delta_{\{v,w\}}+|O_{v,w}|\}$.
\end{corollary}
\begin{proof}
Let $s_0<s\in \N$, $\{v,w\}\in E$, and w.l.o.g.\ $\Ll_{\{v,w\}}(t)=L_w(t)-L_v(t)$.
By \Cref{lem:valid},
\begin{equation*}
\Ll_{\{v,w\}}(t)=L_w(t)-L_v(t)\le \Psi_v^s(t)+d^s(v,w)
\le \Psi^s(t)+\omega^s(v,w)
\le \Psi^s(t)+4s\delta_{\{v,w\}}+|O_{v,w}|.\qedhere
\end{equation*}
\end{proof}
When increasing the level $s$, the potential function can only decrease, as we substract more.
\begin{lemma}\label{lem:psi_half_bound}
For each $s_0<s<s'$, $s,s'\in 1/2\cdot \N$, $v\in V$, and time $t$, we have that $\Psi^{s'}_v(t)\le \Psi^s_v(t)$.
\end{lemma}
\begin{proof}
By definition, for each edge $(v,w)\in \vec{E}$, we have that $\omega^{s'}(v,w)\ge \omega^s(v,w)$.
Hence, $d^{s'}(v,w)\ge d^s(v,w)$ for each $v,w\in V$ and the claim readily follows from \Cref{def:potential}.
\end{proof}

\subsection*{Bounding the Potentials}

As our first step, we bound how quickly potentials can grow.
\begin{lemma}\label{lem:waitup}
Let $s_0\le s\in \N$, $v\in V$, and $t'>t$.
If $\Psi_v^s(\tau)>0$ for all $\tau\in(t,t')$, then
\begin{equation*}
\Psi_v^s(t')\le \Psi_v^s(t)+\vartheta(t'-t)-(L_v(t')-L_v(t)).
\end{equation*}
\end{lemma}
\begin{proof}
By \Cref{lem:implements}, with the exception of the discrete points in time when $r_v$ changes at some node, $\Psi_v^s$ is differentiable with derivative
\begin{equation*}
\frac{d\Psi_v^s}{dt}(t')=\max_{V\ni w \mathrm{~maximizes~}\Psi_v^s(t')}\left\{\frac{dL_w}{dt}(t')-\frac{dL_v}{dt}(t')\right\}.
\end{equation*}
Hence, the claim follows from integration, provided that we can show that for any $w$ maximizing $\Psi_v^s(\tau)>0$ at a given time $\tau\in (t,t')\subseteq I_k$ when $L_w$ is differentiable, it holds that $\frac{dL_w}{dt}(\tau)\le \frac{dH_w}{dt}(\tau)\le \vartheta$.
By \Cref{lem:implements,lem:mutex}, this follows if such a $w$ satisfies the slow trigger.
Therefore, by \Cref{lem:implies} it is sufficient to prove that for any $w$ maximizing $\Psi_v^s(\tau)$ at such a time $\tau$, $w$ satisfies the slow condition, as it implies the slow trigger.

Accordingly, suppose that
\begin{equation*}
\Psi_v^s(\tau)=L_w(\tau)-L_v(\tau)-d^s(v,w)>0,
\end{equation*}
i.e., $w$ maximizes $\Psi_v^s(\tau)>0$.
For any $\{w,x\}\in E$, we have that $L_x(\tau)-L_v(\tau)-d^s(v,x)\le L_w(\tau)-L_v(\tau)-d^s(v,w)$, yielding by the substractive triangle inequality that
\begin{equation*}
L_x(\tau)-L_w(\tau)\le d^s(v,x)-d^s(v,w)\le d^s(w,x)\le \omega^s(w,x) = 4s\delta_{\{w,x\}}-O_{w,x},
\end{equation*}
which can be rearranged to show the second requirement of the slow condition for node $w$ at time $\tau$.

To show the first, we note that $w\neq v$, because $\Psi^s_v(\tau)>0=L_v(\tau)-L_v(\tau)-d^s(v,v)$.
Therefore, we may consider a neighbor $x$ of $w$ on a shortest path from $v$ to $w$ in $G^s$, satisfying $d^s(v,w)=d^s(v,x)+\omega^s(x,w)$.
Using that $w$ maximizes $\Psi^s_v(\tau)$ once more, it follows that
\begin{equation*}
L_w(\tau)-L_x(\tau)\ge d^s(v,w)-d^s(v,x)=\omega^s(x,w)=4s\delta_{\{w,x\}}-O_{x,w}=4s\delta_{\{w,x\}}+O_{w,x},
\end{equation*}
which can be rearranged to show the first requirement.
\end{proof}

\begin{corollary}\label{cor:waitup}
For all $s_0\le s\in \N$, $v\in V$, and $t'>t$, it holds that
\begin{equation*}
\Psi_v^s(t')\le \Psi_v^s(t)+(\vartheta-1)(t'-t).
\end{equation*}
\end{corollary}
\begin{proof}
Follows from \Cref{lem:waitup}, as $\frac{dL_v}{dt}\ge 1$ (except at the discrete points in time when it does not exist) by \Cref{cor:rates}.
\end{proof}

\begin{lemma}\label{lem:catchup}
For each $s_0<s\in \N$, node $v\in V$, and times $t<t'$, we have that
\begin{equation*}
L_v(t')-L_v(t) \ge t'-t+\min\left\{\Psi^{s-1/2}_v(t),\mu(t'-t)-\Psi^{s-1/2}(t)+\Psi^{s-1/2}_v(t)\right\}.
\end{equation*}
\end{lemma}
\begin{proof}
Choose $w\in V$ such that $\Psi_v^{s-1/2}(t)=L_w(t)-L_v(t)-d^{s-1/2}(v,w)$.
For $x\in V$, define
\begin{equation*}
\phi_x(\tau):=L_w(t)+(\tau-t)-L_x(\tau)-d^{s-1/2}(x,w)
\end{equation*}
and let $\Phi(\tau):=\max_{x\in V}\{\phi_x(\tau)\}$.
Note that
\begin{equation*}
L_v(t')-L_v(t)-(t'-t)=\phi_v(t)-\phi_v(t')\ge \phi_v(t)-\Phi(t').
\end{equation*}
Hence, if $\Phi(t')\le 0$, we get that
\begin{equation*}
L_v(t')-L_v(t)\ge t'-t+\phi_v(t)=t'-t+L_w(t)-L_v(t)-d^{s-1/2}(v,w)=t'-t+\Psi^{s-1/2}_v(t),
\end{equation*}
as desired.

Next, observe that by \Cref{cor:rates}
\begin{equation*}
\phi_x(\tau)-\phi_x(t)=\tau-t-(L_x(\tau)-L_x(t))\le 0,
\end{equation*}
i.e., $\phi_x$ is decreasing.
Thus, if $\Phi(\tau)\le 0$ for \emph{any} $\tau\in [t,t']$, the statement of the lemma follows.

Hence, it remains to consider the case that $\Phi(\tau)>0$ for all $\tau\in [t,t']$.
We claim that this entails that $\Phi(t')\le \Phi(t)-\mu(t'-t)$, from which the statement follows due to
\begin{align*}
L_v(t')-L_v(t)-(t'-t)&\ge \phi_v(t)-\Phi(t')\\
&\ge \phi_v(t)-\Phi(t)+\mu(t'-t)\\
&=\Psi^{s-1/2}_v(t)-\max_{x\in V}\{L_w(t)-L_x(t)-d^{s-1/2}(x,w)\}+\mu(t'-t)\\
&\ge\Psi^{s-1/2}_v(t)-\Psi^{s-1/2}(t)+\mu(t'-t).
\end{align*}

To show this claim, note that by \Cref{lem:implements} $\Phi(\tau)$ is
differentiable with derivative
\begin{equation*}
\frac{d\Phi}{d\tau}(\tau)=\max_{V\ni x \mathrm{~maximizes~}\Phi(\tau)}\left\{1-\frac{dL_x}{dt}(\tau)\right\}
\end{equation*}
except for a discrete set of points in time when not all derivatives exist.
Accordingly, consider a node $x\in V$ maximizing $\Phi(\tau)$ for $\tau\in [t,t']$ (when the derivatives exist), i.e., $\phi_x(\tau)=\Phi(\tau)>0$.
We will show that $\frac{dL_x}{dt}(\tau)\ge 1+\mu$, which will establish that $\frac{d\Phi}{d\tau}(\tau)\le -\mu$.
The claim (and hence the lemma) then readily follow from integration over the interval $[t,t']$.
  
To prove that $\frac{dL_x}{dt}(\tau)\ge 1+\mu$, we show that $x$ satisfies the fast condition at time $\tau$.
By \Cref{lem:implies}, the fast trigger then holds at $x$, which by \Cref{lem:implements} yields that indeed $\frac{dL_x}{dt}(\tau)=(1+\mu)\frac{dH_x}{dt}(\tau)\ge 1+\mu$.

To this end, first consider any $\{x,y\}\in E$. As $\phi_y(\tau)\le \phi_x(\tau)$, by the triangle inequality we have that
\begin{align*}
L_x(\tau)-L_y(\tau)&\le d^{s-1/2}(y,w)-d^{s-1/2}(x,w)\\
&\le d^{s-1/2}(y,x)\\
&\le \omega^{s-1/2}(y,x)\\
&=(4s-2)\delta_{\{y,x\}}-O_{y,x}\\
&=(4(s-1)+2)\delta_{\{x,y\}}+O_{x,y},
\end{align*}
which can be rearranged into the second requirement of the fast condition (on level $s-1$) for $x$ at time $\tau$.

  To show the first requirement, note that $x\neq w$, as
\begin{equation*}
\phi_w(\tau)\le \phi_w(t)=0<\Phi(\tau)=\phi_x(\tau).
\end{equation*}
Hence, we may consider the first node $y\in V$ that is on a shortest path from $x$ to $w$, entailing that $d^{s-1/2}(x,w)=\omega^{s-1/2}(x,y)+d^{s-1/2}(y,w)$.
We get that
\begin{equation*}
L_x(\tau)-L_y(\tau)\le d^{s-1/2}(y,w)-d^{s-1/2}(x,w)=-\omega^{s-1/2}(x,y)=-(4(s-1)+2)\delta_{\{x,y\}}+O_{x,y},
\end{equation*}
which can be rearranged into the first requirement of the fast condition on level $s-1$.
\end{proof}

\begin{corollary}\label{cor:catchup}
For each $s_0<s\in \N$, node $v\in V$, and times $t,t'$ with $t'\ge t+\Psi^{s-1/2}(t)/\mu$, it holds that
\begin{equation*}
L_v(t')-L_v(t)\ge t'-t+\Psi^s_v(t).
\end{equation*}
\end{corollary}
\begin{proof}
By \Cref{lem:psi_half_bound} and the lower bound on $t'$,
\begin{equation*}
\Psi_v^s(t)\le \Psi_v^{s-1/2}(t)\le \mu(t'-t)-\Psi^{s-1/2}(t)+\Psi_v^{s-1/2}(t),
\end{equation*}
so the statement follows by application of \Cref{lem:catchup}.
\end{proof}

\Cref{lem:waitup} and \Cref{lem:catchup} are key to obtaining bounds on the potentials that decrease exponentially in $s$.
First, we cover the base case for anchoring the induction at level $s>s_0$.
\begin{lemma}\label{lem:psibase}
Suppose that $s_0<s\in \N$, $\sigma:=\mu/(\vartheta-1)\ge 2$ and that $\varepsilon>0$.
Then at all times $t \in [a+(2\Psi^s(a)+\lceil\log_{\sigma}2/\varepsilon\rceil\cdot 4\W^s)/\mu,a+T]$, it holds that $\Psi^s(t)\le 2\W^s/(\sigma-1)+\varepsilon \Psi^{s-1/2}(0)$.
\end{lemma}
\begin{proof}
Set $t_0:=a$, $t_1:=a+(\Psi^s(a)+2\W^s)/\mu$, and $B_0:=2\Psi^s(a)+2\W^s/\sigma$.
We apply \Cref{cor:catchup} to see that for $t\in [t_0,t_1]$, it holds that
\begin{equation*}
\Psi^s(t)\le \Psi^s(a)+(\vartheta-1)(t_1-a) \le \Psi^s(a)+\frac{\Psi^s(a)+2\W^s}{\sigma}\le B_0.
\end{equation*}

We now inductively define
\begin{equation*}
t_{i+1}:= t_i+\frac{B_i+2\W^s}{\mu}\quad \mbox{and}\quad B_{i+1}:= \frac{B_i+2\W^s}{\sigma}
\end{equation*}
for $i\in \N_{>0}$ and $i\in \N$, respectively.
We claim that $\Psi^s(t)\le B_i$ for all $t\in[t_i,\min\{t_{i+1},T\}]$ and $i\in \N$, which we show by induction on $i$.
We already covered the base case of $i=0$, so assume that the claim holds for index $i\in \N$.
For $t\in [t_{i+1},t_{i+2}]$, choose $v\in V$ such that $\Psi^s(t)=\Psi_v^s(t)$.
Define $t':=\max\{t-(B_i+2\W^s)/\mu,t_0\} \in [t_i,t_{i+1}]$.
By \Cref{cor:global_bound} and the induction hypothesis, 
\begin{equation*}
\Psi^{s-1/2}(t')\le \Psi^s(t')+\W^{s-1/2}+\W^s\le B_i+2\W^s,
\end{equation*}
where the last step uses that $\W^s$ is increasing in $s$, because the same is true for $\omega^s$.
Hence, if $t'=t-(B_i+2\W^s)/\mu$, \Cref{lem:waitup} and \Cref{cor:catchup} state that
\begin{equation*}
\Psi_v^s(t)\le \Psi_v^s(t')+\vartheta(t-t')-(L_v(t)-L_v(t'))
\le (\vartheta-1)(t_{i+1}-t_i)
=\frac{B_i+2\W^s}{\sigma}=B_{i+1}.
\end{equation*}

Note that if $i\ge 1$, then from $t\ge t_{i+1}$ it follows that $t-(B_i+2\W^s)/\mu\ge t_i$ and thus $t':=t-(B_i+2\W^s)/\mu$.
Therefore, the remaining case is that $t'=t_0=a$ and $i=0$.
Then \Cref{cor:global_bound} entails that 
\begin{equation*}
\frac{\Psi^{s-1/2}(t')}{\mu}\le \frac{\Psi^s(a)+\W^{s-1/2}+\W^s}{\mu}\le \frac{\Psi^s(a)+2\W^s}{\mu}= t_1-t_0\le t-t_0.
\end{equation*}
Therefore, this case analogously yields that $\Psi_v^s(t)\le B_1=B_{i+1}$, as desired.

With the induction complete, another straightforward induction on $i$ shows that
\begin{equation*}
B_i=\frac{2\Psi^s(a)}{\sigma^i}+2\W^s\cdot \sum_{j=1}^{i+1}\sigma^{-j}<\frac{2\Psi^s(a)}{\sigma^i}+\frac{2\W^s}{\sigma-1}
\end{equation*}
for all $i\in \N$.
As for $i\ge i_0:=\lceil\log_{\sigma}2/\varepsilon\rceil$, we have that $2\Psi^s(a)/\sigma^i\le \varepsilon \Psi^s(a)$, it remains to show that $t_{i_0}$ is sufficiently small.
To see this, observe that
\begin{align*}
t_{i_0}-a&=\sum_{i=0}^{i_0-1}t_{i+1}-t_i\\
&=\sum_{i=0}^{i_0-1}\frac{B_i+2\W^s}{\mu}\\
&<\sum_{i=0}^{i_0-1}\left(\frac{\Psi^s(a)}{\mu\sigma^i}+\frac{2\W^s}{\mu}\cdot \sum_{j=0}^{i}\sigma^{-j}\right)\\
&<\frac{\sigma}{\sigma-1}\cdot\frac{\Psi^s(a)+i_0\cdot 2\W^s}{\mu}\\
&\le \frac{2\Psi^s(a)+\lceil\log_{\sigma}2/\varepsilon\rceil\cdot 4\W^s}{\mu},
\end{align*}
where the final step again uses that $\sigma\ge 2$.
\end{proof}
\begin{corollary}\label{cor:psibase}
Suppose that $s_0<s\in \N$, $\sigma:=\mu/(\vartheta-1)\ge 2$, $\varepsilon>0$ is a constant, and $T\ge C(\Psi^s(a)+\W^s)/\mu$ for a suitable constant $C=C(\varepsilon)$.
Then at all times $t \in [a+T/4,a+T]$, it holds that $\Psi^s(t)\le 2\W^s/(\sigma-1)+\varepsilon \Psi^{s-1/2}(a)$.
\end{corollary}
Using the preceding lemmas, bounding $\Psi^{s+1}$ given a bound on $\Psi^s$ is straightforward.
\begin{lemma}\label{lem:psistep}
Suppose that for some $s_0<s\in \N$, it holds that $\Psi^s(t')\le B$ at all times $t'\ge t$.
Then $\Psi^{s+1}(t')\le B/\sigma$ at all times $t'\ge t+B/\mu$, where $\sigma=\mu/(\vartheta-1)$.
\end{lemma}
\begin{proof}
For any $t'\ge t+B/\mu$, consider $v\in V$ such that $\Psi^{s+1}(t')=\Psi^{s+1}_v(t')$.
By assumption and \Cref{lem:psi_half_bound}, we have that
\begin{equation*}
t'\ge t'-\frac{B-\Psi^{s+1}(t'-B/\mu)}{\mu}\ge t'-\frac{B-\Psi^{s+1/2}(t'-B/\mu)}{\mu}=t'-\frac{B}{\mu}+\frac{\Psi^{s+1/2}(t'-B/\mu)}{\mu}.
\end{equation*}
Thus, \Cref{lem:waitup} and \Cref{cor:catchup} yield that
\begin{equation*}
\Psi^{s+1}(t')=\Psi^{s+1}_v(t')\le \Psi^{s+1}_v\left(t'-\frac{B}{\mu}\right)+\vartheta\cdot\frac{B}{\mu}-\left(L_v(t')-L_v\left(t'-\frac{B}{\mu}\right)\right)\le (\vartheta-1)\cdot\frac{B}{\mu}=\frac{B}{\sigma}.\qedhere
\end{equation*}
\end{proof}
Putting the induction together, we can bound the potentials and, by extension, skews.
\begin{theorem}\label{thm:psi}
Suppose that $s_0<s\in \N$, $\sigma:=\mu/(\vartheta-1)\ge 2$, $\varepsilon>0$ is a constant, and $T\ge C(\Psi^s(a)+\W^s)/\mu$ for a suitable constant $C=C(\varepsilon)$.
Then for each $s\le s'\in \N$ and time $t \in [a+T/2,a+T]$, it holds that
\begin{equation*}
\Psi^{s'}(t)\le \frac{2\W^s/(\sigma-1)+\varepsilon \Psi^{s-1/2}(a)}{\sigma^{s'-s}}.
\end{equation*}
\end{theorem}
\begin{proof}
Abbreviate $\psi:=2\W^s/(\sigma-1)+\varepsilon \Psi^{s-1/2}(a)$.
By \Cref{cor:psibase}, $\Psi^s(t)\le \psi$ at times $t \in [a+T/4,a+T]$.
Inductive application of \Cref{lem:psistep} establishes the claimed bound for each $s'$ and times larger than
\begin{equation*}
a+\frac{T}{4}+\sum_{s''=s+1}^{s'} \frac{\psi}{\mu\sigma^{s''-s}}<a+\frac{T}{4}+ \frac{\sigma\psi}{\mu(\sigma-1)}\le a+\frac{T}{4}+ \frac{2\psi}{\mu}\le a+\frac{T}{2},
\end{equation*}
where the last two steps use that $\sigma\ge 2$ and that $C$ is sufficiently large.
\end{proof}

\begin{corollary}\label{cor:psi}
Suppose that $s_0<s\in \N$, $\sigma:=\mu/(\vartheta-1)\ge 2$, $\varepsilon>0$ is constant, and $T\ge C(\Psi^s(a)+\W^s)/\mu$ for a suitable constant $C$.
Then for each $s\le s'\in \N$ and time $t \in [a+T/2,a+T]$, it holds that
\begin{align*}
\Psi^{s'}(t)&\le \frac{(2+\varepsilon)\W^s}{\sigma^{s'-s}(\sigma-1)},\\
\G(t) &\le \W^s+\frac{(2+\varepsilon)\W^s}{\sigma-1},\mbox{ and}\\
\Ll_{\{v,w\}}(t) &\in |O_{v,w}|+4\left(s+\log_{\sigma}\left(\frac{\W^s}{\delta_{\{v,w\}}}\right)+\BO(1)\right)\delta_{\{v,w\}}\mbox{ for each }\{v,w\}\in E. 
\end{align*}
\end{corollary}
\begin{proof}
Follows from \Cref{thm:psi}, \Cref{cor:global_bound}, and \Cref{cor:local_bound}, where we plug in $s'=s+\lceil\log_{\sigma}\W^s/\delta_{\{v,w\}}\rceil$.
\end{proof}
\begin{proof}[Proof of \Cref{cor:uniform}]
Set $s:=\lceil \Delta/(4\delta(T))\rceil$.
As the graph is uniform with edge parameters $\Delta$ and $\delta:=\delta(T)$, \Cref{lem:s0} states that $s_0\in \Delta/(4\delta)+\BO(1)$.
Hence, $\W^s\in D(\Delta+\BO(\delta))$.
By \Cref{cor:waitup,cor:global_bound}, $\Psi^s(C\W^s/(2\mu))\le \Psi^s(0)+(\vartheta-1)C\W^s/(2\mu)\le \G(0)+\BO(\W^s)\in \BO(D\Delta)$.

We inductively apply \Cref{cor:psi} with\footnote{This requires that $C\W^s/(2\mu)\ge 1$. However, if this is not the case, we may modify \Cref{alg:gcs_basic} to trigger computational steps whenever $H_v(t)\bmod C\W^s/(2\mu)\cdot\N=0$.} $a=C\W^s/(2\mu)+iT/2$ for $i\in \N$, $\varepsilon=1$, and $s=s_0+1$, yielding the desired bound on the global skew at times $t\ge T\ge C\W^s/(2\mu)+T/2$.
Using that $\sigma\ge 2$, the corollary then bounds the local skew by
\begin{equation*}
2\Delta+4\delta\left(\log_{\sigma}\frac{D\Delta}{\delta}+\BO(1)\right)=2\Delta+4\delta\left(\log_{\sigma}D+\log_{\sigma}\frac{\Delta}{4\delta}+\BO(1)\right)\le 3\Delta+4\delta\left(\log_{\sigma}D+\BO(1)\right)
\end{equation*}
at such times.
\end{proof}

\section{Bounding Stabilization Time}\label{sec:stab}

The only piece missing for \Cref{alg:gcs_basic} to become self-stabilizing is to control the ``initial'' value of $\Psi^{s_0+1}$.
Fortunately, we can follow the typical ``detect and reset'' approach to self-stabilization here, in that either $\Psi^{s_0+1}\in \BO(\W^{s_0+1})$ or we can detect this and re-initialize the logical clocks so that offsets satisfy a sufficiently small bound.

\subsection*{The Idea}
For the sake of the argument, pretend that we could take an instantaneous snapshot of all $o_{\{v,w\}}(t)$, $\{v,w\}\in E$, at some time $t$.
Recall that $o_{v,w}(t)=L_v(t)-L_w(t)-O_{v,w}\pm \delta_{\{v,w\}}$ and $O_{v,w}=-O_{w,v}$.
For a walk $v_0,\ldots,v_{\ell}$, we thus have that
\begin{align*}
\sum_{i=1}^{\ell}o_{v_{i-1},v_i}(t)&=L_{v_0}(t)-L_{v_{\ell}}(t)-\sum_{i=1}^{\ell}O_{v_{i-1},v_i}\pm \sum_{i=1}^{\ell}\delta_{\{v_{i-1},v_i\}}\\
&=L_{v_0}(t)-L_{v_{\ell}}(t)+\sum_{i=1}^{\ell}O_{v_i,v_{i-1}}\pm \sum_{i=1}^{\ell}\delta_{\{v_i,v_{i-1}\}}.
\end{align*}
Up to the error term of $\pm\sum_{i=1}^{\ell}\delta_{\{v_i,v_{i-1}\}}$ and the need to add $\sum_{i=1}^{\ell}4s\delta_{\{v_i,v_{i-1}\}}$, this is exactly the contribution of the ``reversed'' walk $v_{\ell},\ldots,v_0$ to $\Psi^s(t)$.
Denoting for $s\in 1/2\cdot \N$
\begin{equation*}
\tilde{\Psi}^{s+1/4}(t):=\sup_{v,w\in V}\left\{\sup_{v_0,\ldots,v_{\ell} \in W_{v,w}}\left\{\sum_{i=1}^{\ell}o_{v_{i-1},v_i}(t)-(4s+1)\delta_{\{v_{i-1},v_i\}}\right\}\right\},
\end{equation*}
we thus have the following.
\begin{itemize}
  \item $\Psi^{s+1/2}(t)\le \tilde{\Psi}^{s+1/4}(t)\le \Psi^s(t)$.
  \item Denote by $\tilde{s}_0$ the minimum $s\in 1/2\cdot \N$ such that $\tilde{\Psi}^{s+1/4}(t)$ is finite, i.e., the graph $(V,\vec{E},(4s+1)\delta_e-o_e(t))$ has no negative cycles. Then $\tilde{s}_0\in \{s_0,s_0+1/2\}$, i.e., $s_0=\lfloor \tilde{s}_0\rfloor$.
  \item As for $s_0<s\in \N$ we have that $|\Psi^{s-1/2}(t)-\Psi^s(t)|\le \W^{s-1/2}+\W^s\le 2\W^s$, by \Cref{cor:global_bound} $\Psi^{s_0+1}(t)-2\W^{s_0+1}\le \tilde{\Psi}^{s_0+3/4}(t)\le \Psi^{s_0+1}(t)$.
\end{itemize}
All in all, this would enable us to estimate $\Psi^{s_0+1}(t)$ up to an error of $\W^{s_0+1}$, which is sufficient for our purposes.

\subsection*{Taking a Snapshot}
The main obstacle to this strawman approach is that we cannot capture all $o_{(v,w)}$ at the same time instant $t$.
We fall back to having nodes note down these values when they are reached by a flooding initiated by the (pre-defined) root $r$.
In addition, we store the logical times at which these values are captured.
Afterwards, $r$ collects all values and uses them to determine $\tilde{s}_0\in (s_0,s_0+\BO(1)]$ and estimate $\Psi^{\tilde{s}_0}$.
If $\Psi^{\tilde{s}_0}$ is too large, $r$ initializes a reset of the logical clocks.
To ensure timely self-stabilization, $r$ controls this entire process as the root of a shortest-path tree with respect to edge weights $d_e$, the maximum communication delay of edge $e$.
\begin{definition}\label{def:diam_delay}
By $\Wd$ we denote the diameter of the graph $(V,E,d_e)$, i.e., the minimum time that is sufficient to enable all nodes to communicate each other when forwarding messages.
\end{definition}
First, we observe that we can construct and maintain a self-stabilizing shortest-path tree quickly.
Intuitively, the construction follows a Bellman-Ford approach, with the known delay bounds enabling us to ``flush out'' incorrect distance information.
\begin{lemma}\label{lem:bfs}
A shortest-path tree of $(V,E,d_e)$ rooted at $r$ can be constructed with stabilization time $2\Wd$.
\end{lemma}
\begin{proof}
Non-root nodes maintain for each neighbor a variable indicating their distance to the root via this neighbor, as well as which neighbor they believe to be their parent.
A variable corresponding to edge $e$ cannot attain smaller values than $d_e$.
The root will send a $0$-message along each of its incident edges $e$ every $d_e$ time according to its hardware clock.

For every incident edge $e$, every $d_e$ time each non-root node sends the minimum of its current distance variables along $e$.
Whenever receiving a value $d$ on edge $\{v,w\}$, non-root node $v$ updates its distance variable for this edge to $d+d_{\{v,w\}}$ and sets its parent variable to the neighbor corresponding to the variable providing this minimum (following some consistent tie-breaking rule in case of equality).

To see that these rules indeed result in the parent variables stabilizing to those of a shortest-path tree within $\BO(\Wd)$ time, we perform an induction over the correct distances $d_v$, $v\in V$, from the root.
The claim is that at times $t\ge 2d_v$, (i) distance variables of node $w\in V$ are at least $\min\{d_w,d_v\}$ and (ii) if $d_w\le d_v$, then there is a distance variable equalling $d_w$ that corresponds to the parent of $w$ in the tree.
Note that from this the claim of the lemma is immediate, since all nodes are within distance $\Wd$ from the root.

We will prove this by induction over the finitely many distances of nodes from the root.
The induction is anchored at $d_v=0$, since the root always has correct values by construction.
For the step, suppose that the induction hypothesis holds and consider a node $v\in V$ in distance $d_v$ from the root.
For any node $w$ and edge $\{w,x\}\in E$, suppose that $w$ receives a message at time $t\ge 2d_v$ or later.
As $x$ sends a message over $\{w,x\}$ at least every $d_{\{w,x\}}$ time, this message was sent no earlier than time $t-2d_{\{w,x\}}$.
By the induction hypothesis, it then either contained the correct distance of $x$ to the root or a value of at least $d_v-d_{\{w,x\}}$.
In the former case, such a message will not cause $w$ to set its distance variable for $\{w,x\}$ to a smaller value than its own distance the root;
in the latter case, such a message will not cause $w$ to set its distance variable for $\{w,x\}$ to a smaller value than $d_v$.

Claim (i) follows, since it is trivial if $t\le 2d_{\{w,x\}}$, and $w$ receives its first message from $x$ at the latest by time $2d_{\{w,x\}}$.
For claim (ii), suppose that $x$ is the true parent of $w$ in the tree.
Then, by the induction hypothesis, the minimum of its distance variables is equal to $d_x$ at times $t\ge 2d_x$.
Thus, the latest message $w$ received from $x$ by time $t\ge 2d_w$, which was sent no earlier than time $t-2d_{\{x,w\}}\ge 2d_w-2d_{\{x,w\}}=2d_x$, by the induction hypothesis carries value $d_x$.
Hence $w$ sets its corresponding variable to $d_x+d_{\{x,w\}}=d_w$.
\end{proof}

With the tree in place, it is trivial to add a self-stabilizing computation of the weighted depth of the tree with stabilization time $\BO(\Wd)$, which then can be distributed to all nodes with stabilization time $\BO(\Wd)$.
This depth is a factor-$2$ approximation of $\Wd$.
Using this information, every $\BO(\Wd)$ time the root can schedule, via the tree, an execution of the following routine using its local time, guaranteeing that there is no overlap between steps.
\begin{enumerate}
  \item Clear/initialize all nodes' state for the routine via the tree ($\BO(\Wd)$ time).
  \item Perform a flooding to trigger each node $v\in V$ capturing the values $o_{v,w}(t_v)$ for each $\{v,w\} \in E$ at a time $t_v$, on the subgraph induced by all edges whose weights do not exceed twice the depth of the tree. Note that this ensures for any $\{v,w\}\in E$ that $|t_v-t_w|\le d_{\{v,w\}}$, either because the edge is part of the subgraph or its delay exceeds $\Wd$, which is an upper bound on the time required for the flooding. This completes within $\BO(\Wd)$ time, in the sense that no flooding messages remain in transit on any link of this subgraph.
  \item For each $\{v,w\}\in E$, the root learns $o_{v,w}(t_v)$, $o_{w,v}(t_w)$, and $\delta_{v,w}$. This is done by means of a convergecast on the tree ($\BO(\Wd)$ time).
  \item The root determines whether a reset of the logical clocks is required (see below) and, if so, triggers it by another flooding on the same subgraph ($\BO(\Wd)$ time).
\end{enumerate}
To fully specify the above routine, it remains to discuss how the root can decide, based on the collected values, whether a reset is necessary, and how it is executed.

\subsection*{Estimating \texorpdfstring{$\Psi$}{Psi}}
To this end, we first show how the root can estimate $\Psi^{\tilde{s}_0+1}(t_r)$.
This estimate will be expressed in terms of the diameter of $G$ with respect to distances $\delta_e$.
\begin{definition}[Estimate Distance]
We denote by $d^{\delta}$ the distance function of $(V,E,\delta_e)$ and by $\Wdelta$ its diameter, i.e., $\max_{v,w\in V}\{d^{\delta}(v,w)\}$.
\end{definition}
\begin{lemma}\label{lem:est_Psi}
After the above routine is complete, the root $r$ can compute a value $\psi$ such that
\begin{equation*}
\Psi^s(t_r)\le \psi \le \Psi^{s_0+1/2}(t_r)+\BO(\Wdelta)
\end{equation*}
for some $s_0<s\in s_0+\BO(1)$.
\end{lemma}
\begin{proof}
Note that for any $v\in V$, and times $t$, $t'$, it holds that
\begin{equation*}
L_v(t')-L_v(t)=t'-t\pm\max\{|\beta-1|,|\alpha-1|\}|t'-t|=t'-t\pm(\beta-\alpha)|t'-t|.
\end{equation*}
Consider any $v,w\in V$ and walk $W=v_0,\ldots,v_{\ell}\in W_{v,w}$.
We have that
\begin{align*}
&\,\sum_{i=1}^{\ell}\frac{o_{v_i,v_{i-1}}(t_{v_i})-o_{v_{i-1},v_i}(t_{v_{i-1}})}{2}\\
=&\,\sum_{i=1}^{\ell}\frac{L_{v_i}(t_{v_i})-L_{v_{i-1}}(t_{v_i})-(L_{v_{i-1}}(t_{v_{i-1}})-L_{v_i}(t_{v_{i-1}}))-e_{v_i,v_{i-1}}(t_{v_i})+e_{v_{i-1},v_i}(t_{v_{i-1}})}{2}
\end{align*}
\begin{align*}
=&\,\sum_{i=1}^{\ell}\frac{L_{v_i}(t_{v_i})-L_{v_{i-1}}(t_{v_i})-(L_{v_{i-1}}(t_{v_{i-1}})-L_{v_i}(t_{v_{i-1}}))}{2}-O_{v_i,v_{i-1}}\pm \frac{\delta_{v_i,v_{i-1}}}{2}\\
=&\,\sum_{i=1}^{\ell}\frac{L_{v_i}(t_{v_i})-L_{v_{i-1}}(t_{v_i})-(L_{v_{i-1}}(t_{v_{i-1}})-L_{v_i}(t_{v_{i-1}}))}{2}-O_{v_i,v_{i-1}}\pm \frac{\delta(W)}{2}\\
=&\,\sum_{i=1}^{\ell}L_{v_i}(t_{v_i})-L_{v_{i-1}}(t_{v_{i-1}})+\frac{L_{v_i}(t_{v_{i-1}})-L_{v_i}(t_{v_i})+L_{v_{i-1}}(t_{v_{i-1}})-L_{v_{i-1}}(t_{v_i})}{2}\\
&\qquad +O_{v_i,v_{i-1}}\pm \frac{\delta(W)}{2}\\
=&\,\sum_{i=1}^{\ell}L_{v_i}(t_{v_i})-L_{v_{i-1}}(t_{v_{i-1}})-(t_{v_i}-t_{v_{i-1}})-O_{v_i,v_{i-1}}\pm \frac{\delta(W)}{2}+\sum_{i=1}^{\ell}(\beta-\alpha)|t_{v_i}-t_{v_{i-1}}|\\
=&\,\sum_{i=1}^{\ell}L_{v_i}(t_{v_i})-L_{v_{i-1}}(t_{v_{i-1}})-(t_{v_i}-t_{v_{i-1}})-O_{v_i,v_{i-1}}\pm \frac{\delta(W)}{2}+\sum_{i=1}^{\ell}(\beta-\alpha)d_{v_i,v_{i-1}}\\
=&\,\sum_{i=1}^{\ell}L_{v_i}(t_{v_i})-L_{v_{i-1}}(t_{v_{i-1}})-(t_{v_i}-t_{v_{i-1}})-O_{v_i,v_{i-1}}\pm \left(c+\frac{1}{2}\right)\delta(W)\\
=&\,L_w(t_w)-L_v(t_v)-(t_w-t_v)+\sum_{i=1}^{\ell}O_{v_{i-1},v_i}\pm \left(c+\frac{1}{2}\right)\delta(W),
\end{align*}
where we used that $O_{v_i,v_{i-1}}=-O_{v_{i-1},v_i}$.
In particular, if $v=w$, we have that
\begin{equation*}
\sum_{i=1}^{\ell}\frac{o_{v_i,v_{i-1}}(t_{v_i})-o_{v_{i-1},v_i}(t_{v_{i-1}})}{2}=\sum_{i=1}^{\ell}O_{v_{i-1},v_i}\pm \left(c+\frac{1}{2}\right)\delta(W).
\end{equation*}
For $(v,w)\in \vec{E}$, let $\omega_{(v,w)}:= (o_{v,w}(t_v)-o_{w,v}(t_w))/2$.
Let $\tilde{s}_0\in \N$ be minimum with the property that $(V,\vec{E},4\tilde{s}_0\delta_e-\omega_e)$ has no negative cycle.
By the above bound, then $\tilde{s}_0=s_0\pm c+1$.

Fix $s:=\tilde{s}_0+c+2$, i.e., $s_0< s\le s_0+2c+3\in s_0+\BO(1)$.
By the above calculation, we have for each $v,w\in V$ and path $W=v_0,\ldots,v_{\ell}\in W_{v,w}$ that
\begin{align*}
&\,\sum_{i=1}^{\ell}\frac{o_{v_i,v_{i-1}}(t_{v_i})-o_{v_{i-1},v_i}(t_{v_{i-1}})}{2}\\
=\,&L_w(t_w)-L_v(t_v)-(t_w-t_v)+\sum_{i=1}^{\ell}O_{v_{i-1},v_i}\pm \left(c+\frac{1}{2}\right)\delta(W)\\
=\,&L_w(t_r)-L_v(t_r)+\sum_{i=1}^{\ell}O_{v_{i-1},v_i}\pm \left(c+\frac{1}{2}\right)\delta(W)+(\beta-\alpha)(|t_v-t_r|+|t_w-t_r|)\\
=\,&L_w(t_r)-L_v(t_r)+\sum_{i=1}^{\ell}O_{v_{i-1},v_i}\pm \left(c+\frac{1}{2}\right)\delta(W)+2(\beta-\alpha)\Wd\\
=\,&L_w(t_r)-L_v(t_r)+\sum_{i=1}^{\ell}O_{v_{i-1},v_i}\pm \left(c+\frac{1}{2}\right)\delta(W)+2\Wdelta.
\end{align*}
Hence, for each path $W\in W_{v,w}$ we have that
\begin{align*}
&\,L_w(t_r)-L_v(t_r)-\sum_{i=1}^{\ell}(4(s_0+2c+3)\delta_{v_{i-1},v_i}-O_{v_{i-1},v_i})-2\Wdelta\\
\le&\, \sum_{i=1}^{\ell}\frac{o_{v_i,v_{i-1}}(t_{v_i})-o_{v_{i-1},v_i}(t_{v_{i-1}})}{2}-4(s_0+c+2)\delta_{v_{i-1},v_i}\\
\le&\,L_w(t_r)-L_v(t_r)-\sum_{i=1}^{\ell}(4(s_0+1)\delta_{v_{i-1},v_i}-O_{v_{i-1},v_i})+2\Wdelta
\end{align*}
We conclude that
\begin{align*}
&\,\Psi^{s_0+2c+3}(t_r)-2\Wdelta\\
\le\,& \sup_{v,w\in V}\left\{\sup_{v_0,\ldots,v_{\ell} \in W_{v,w}}\left\{\sum_{i=1}^{\ell}\frac{o_{v_i,v_{i-1}}(t_{v_i})-o_{v_{i-1},v_i}(t_{v_{i-1}})}{2}-4(s_0+c+2)\delta_{v_{i-1},v_i}\right\}\right\}\\
\le\,& \Psi^{s_0+1}(t_r)+2\Wdelta.
\end{align*}
Noting that $(V,\vec{E},4(s_0+c+2)\delta_{v_{i-1},v_i}-\omega_e)$ has no negative cycles either, the root can compute the inner supremum as the negated distance $-d(v,w)$ from $v$ to $w$ in $(V,\vec{E},4(s_0+c+2)\delta_e-\omega_e)$, i.e., $\Psi^{s_0+2c+3}(t_r)-2\Wdelta\le \max_{(v,w)\in V\times V}\{-d(v,w)\}\le \Psi^{s_0+1}(t_r)+2\Wdelta$.
Choosing $\psi=2\Wdelta+\max_{(v,w)\in V\times V}\{-d(v,w)\}$ yields the claim of the lemma.
\end{proof}
We can bound $\Wdelta$ by $\W^s$ to obtain an estimate of $\Psi^s(t_r)$ for $s_0<s\le s_0+\BO(1)$.
\begin{lemma}\label{lem:Wdelta}
For $s_0<s\in \N$, it holds that $2\Wdelta \le \W^s$.
\end{lemma}
\begin{proof}
Choose $v,w\in V$ such that $d^{\delta}(v,w)=\Wdelta$.
Without loss of generality, suppose that $d^{s_0+1/2}(v,w)\le d^{s_0+1/2}(w,v)$.
Let $P$ be a shortest path from $v$ to $w$ in $G^{s_0+1/2}$ and $Q$ be a shortest path from $w$ to $v$ in $G^s$.
We have that $\omega^s(Q)\ge d^{s_0+1/2}(w,v)\ge 0$, as by the triangle inequality
\begin{equation*}
2d^{s_0+1/2}(w,v)\ge d^{s_0+1/2}(w,v)+d^{s_0+1/2}(v,w)\ge d^{s_0+1/2}(w,w)\ge 0.
\end{equation*}
Hence,
\begin{equation*}
\W^s\ge d^s(w,v)=\omega^{s_0+1/2}(Q)+4(s-s_0-1/2)\delta(Q)\ge 2\Wdelta.\qedhere
\end{equation*}
\end{proof}

\begin{corollary}\label{cor:est_Psi}
After the above routine is complete, the root $r$ can compute a value $\psi$ such that
\begin{equation*}
\Psi^{s}(t_r)\le \psi \le \Psi^{s}(t_r)+\BO(\W^s)
\end{equation*}
for some $s_0<s\in s_0+\BO(1)$.
\end{corollary}
\begin{proof}
Follows from \Cref{lem:est_Psi} by applying \Cref{lem:psi_half_bound,lem:Wdelta} to see that
\begin{equation*}
\Psi^{s_0+1}(t_r)+\BO(\Wdelta)\le\Psi^{s}(t_r)+\BO(\W^s).\qedhere
\end{equation*}
\end{proof}

\subsection*{Performing Resets}
We now have established that when initiating the above routine at time $t_r$, the root can compute the bound $\psi$ of $\Psi^{\tilde{s}_0}(t_r)$ by time $t_r+\BO(\Wd)=t_r+\BO(\Wdelta/(\beta-\alpha))=t_r+\BO(\W^{\tilde{s}_0}/\mu)$, where $\tilde{s}_0\in s_0+\BO(1)$.
We use this to have the root trigger a reset of the logical clocks if this estimate exceeds the bound on $\Psi^{\tilde{s}_0}(t_r)$ our machinery guarantees.
For this approach, the last missing piece is to perform a sufficiently accurate and fast reset.
\begin{lemma}\label{lem:reset}
If the root decides to perform a reset (when completing a properly initialized instance of the above estimation routine) at time $t$, we can ensure that this is completed by time $t+\Wd$ and it holds that $\Psi^{\tilde{s}_0}(t+\Wd)\in \BO(\W^{\tilde{s}_0})$.
\end{lemma}
\begin{proof}
The root $r$ computes an assignment $\ell_v$ of logical clock values for each node $v\in V$ such that $\ell_r=0$ and, if we shifted $L_v$ by $\ell_v$ for each $v\in V$ without other changes, then $\Psi^{\tilde{s}_0}(t_r)\in \BO(\W^{\tilde{s}_0})$.
By \Cref{cor:global_bound}, such an assignment exists: If we had $L_v(t_r)=L_r(t_r)$ for all $v\in V$, then $\Psi^{\tilde{s}_0}(t_r)\in \BO(\W^{\tilde{s}_0})$, so $\ell_v=L_r(t_r)-L_v(t_r)$ achieves this.
Moreover, by \Cref{lem:est_Psi}, $r$ can estimate $L_w(t_r)-L_v(t_r)-d^{\tilde{s}_0}(v,w)$ with error $\BO(\W^{\tilde{s}_0})$, implying that the root can indeed compute such an assignment.\footnote{One might think of this as $r$ being allowed to freely shift $L_v$ for $v\in V$ and changing the estimates $o_{v,w}$, $(v,w)\in \vec{E}$, accordingly. This results in valid error bounds for the fictional scenario in which the clocks are shifted, so the bounds used in the proof of \Cref{lem:est_Psi} apply in the same way.}

The root informs each node $v\in V$ about the computed value $\ell_v$, which is done by relaying these values over the tree.
On reception at time $t_v'$, $v$ resets its logical clock by setting $L_v(t_v')\gets \hat{L}_v(t_v')+\ell_v$, where $\hat{L}_v(t_v')$ denotes the logical clock value of $v$ just before executing the reset.
Note that this process completes by time $t+\Wd$, so it remains to show that $\Psi^{\tilde{s}_0}(t+\Wd)\in \BO(\W^{\tilde{s}_0})$ as a result.

Consider an arbitrary node $v\in V$.
Using that $t-t_r\in \BO(\Wd)$, we see that\footnote{Note that $L_v$ obeys the rate bounds except at time $t_v'$, when it ``jumps'' by $\ell_v$.}
\begin{align*}
L_v(t+\Wd)&=L_v(t_v')+L_v(t+\Wd)-L_v(t_v')\\
&=
\hat{L}_v(t_v')+\ell_v+L_v(t+\Wd)-L_v(t_v')\\
&=L_v(t_r)+\ell_v+\hat{L}_v(t_v')-L_v(t_r)+L_v(t+\Wd)-L(t_v')\\
&=L_v(t_r)+\ell_v+t+\Wd-t_v'+t_v'-t_r\pm (\beta-\alpha)(t+\Wd-t_v'+t_v'-t_r)\\
&=L_v(t_r)+\ell_v+t+\Wd-t_r\pm (\beta-\alpha)(t+\Wd-t_r)\\
&=L_v(t_r)+\ell_v+t+\Wd-t_r\pm \BO((\beta-\alpha)\Wd)\\
&=L_v(t_r)+\ell_v+t+\Wd-t_r\pm \BO(\Wdelta)\\
&=L_v(t_r)+\ell_v+t+\Wd-t_r\pm \BO(\W^{\tilde{s}_0}),
\end{align*}
where the last step applies \Cref{lem:Wdelta}.
For any $v,w\in V$, it follows that
\begin{equation*}
L_w(t+\Wd)-L_v(t+\Wd)=L_w(t_r)+\ell_w-(L_v(t_r)+\ell_v)\pm \BO(\W^{\tilde{s}_0})
\end{equation*}
and, by the choice of $\ell_v$, $v\in V$, hence
\begin{equation*}
\Psi^{\tilde{s}_0}(t+\Wd)\in \BO(\W^{\tilde{s}_0}).\qedhere
\end{equation*}
\end{proof}

\begin{theorem}\label{thm:stab}
Suppose that $\sigma:=\mu/(\vartheta-1)\ge 2$, $\varepsilon>0$ is constant, and $T\ge C\W^s/\mu$ for a suitable constant $C>0$.
Then a variant of \Cref{alg:gcs_basic} satisfies the following.
There is $\tilde{s}_0\in \N$, $s_0<\tilde{s}_0\in s_0+\BO(1)$, such that for each $\tilde{s}_0\le s\in \N$ and time $t \in [a+T/2,a+T]$, it holds that
\begin{align*}
\Psi^s(t)&\le \frac{(2+\varepsilon)\W^{\tilde{s}_0}}{\sigma^{s-\tilde{s}_0}(\sigma-1)},\\
\G(t) &\le \W^{\tilde{s}_0}+\frac{(2+\varepsilon)\W^{\tilde{s}_0}}{\sigma-1},\mbox{ and}\\
\Ll_{\{v,w\}}(t) &\in |O_{v,w}|+4\left(\tilde{s}_0+\log_{\sigma}\left(\frac{\W^{\tilde{s}_0}}{\delta_{\{v,w\}}}\right)+\BO(1)\right)\delta_{\{v,w\}}\mbox{ for each }\{v,w\}\in E. 
\end{align*}
\end{theorem}
\begin{proof}
The root repeatedly executes the estimate procedure discussed above, cf.~\Cref{lem:est_Psi}.
It then triggers a reset if the estimated skew exceeds $(2+\varepsilon)\W^{\tilde{s}_0}/(\sigma-1)$, cf.~\Cref{lem:reset}.
After a reset, it waits for $\BO(\vartheta S+\W^{\tilde{s}_0}/\mu)$ time before re-initializing the next instance of the estimate procedure, where the constants are chosen such that there is sufficient time to accommodate the time bound of $\BO(\W^{\tilde{s}_0}/\mu)$ stated in \Cref{thm:psi} until $\Psi^{\tilde{s}_0}\le (2+\varepsilon)\W^{\tilde{s}_0}/(\sigma-1)$ is guaranteed.
However, note that the clock ``jumps'' caused by a reset invalidate the offset estimates, and these need to recover their guarantees first before we can apply the analysis from \Cref{sec:algo}.
With $S$ being the stabilization time of the estimation procedure, we hence need to wait for $S$ time first, during which $\Psi^{s_0+1/2}$ grows in the worst case by $(\vartheta\mu+(\vartheta-1))S\le(2\vartheta\mu S)$;
this increases the stabilization time of the algorithm by an additive $\vartheta S$, which is factored into the time the root waits.

The proof of stabilization is now straightforward.
Within $\BO(\vartheta S +\W^{\tilde{s}_0}/\mu)$ time, the root initializes a fresh instance of the estimate procedure, clearing any possible corrupted state of this subroutine first.
We distinguish two cases.
First, if this instance triggers a reset, by \Cref{lem:reset} this ensures that $\Psi^{\tilde{s}_0}\in \BO(\W^{\tilde{s}_0})$ after the routine completes.
\Cref{cor:psi} then shows the claimed bounds, where by \Cref{thm:psi} the root waits for long enough before initializing the next instance of the estimate procedure such that it is guaranteed to not trigger a reset, because skews are small enough.

On the other hand, if no reset is triggered, by \Cref{cor:psi} we have that $\Psi^{\tilde{s}_0}\in \BO(\W^{\tilde{s}_0})$.
By \Cref{cor:psi}, the claimed bounds follow, where again the root waits for long enough before initializing the next instance so that no reset is possible in the future.
\end{proof}
\Cref{cor:stab} now follows analogously to \Cref{cor:uniform} from this theorem, noting that replacing $s_0$ by $\tilde{s}_0$ only affects the constants in the $\BO$-notation in this special case.

\section{External Synchronization}\label{sec:external}

Intuitively, our approach slows down the logical clocks of all nodes $v\in V$ by a factor of $\zeta\in (\vartheta,1+\mu)$ and simulates a virtual reference node $v_0$ with $L_{v_0}(t)=H_v(t)= t$ for some $\vartheta<\zeta<1+\mu$ that never satisfies the fast condition on any level $s>s_0$.
This ensures that each $v\in R$ has an estimate of the ``logical clock'' $L_{v_0}(t)=t$ of $v_0$ via its external reference, and can simulate an edge $e=\{v_0,v\}$ to this virtual node with $\delta_e=\delta_v(T)$.
Because $v_0$ ``implements'' the algorithm, in the sense that (i) it satisfies $\frac{dL_{v_0}}{dt}=\frac{dH_{v_0}}{dt}=1$, (ii) the fast condition never holds on any level $s$ we use in our analysis, and (iii) from the perspective of the simulated algorithm, the logical clock rate of $1$ lies in between the minimum and maximum rates of $1/\zeta$ and $\vartheta(1+\mu)/\zeta$, our machinery can be applied to obtain bounds on $|L_{v_0}(t)-L_v(t)|$ for all nodes $v\in V$ and times $t$, resulting in an external clock synchronization algorithm.
However, note that $v_0$, which is always ``slow,'' has a hardware clock running at rate $1$, exceeding the (simulated) nominal hardware clock rate of $1/\zeta$
by a factor of $\zeta$.
Accordingly, we must replace $\vartheta$ by $\zeta$ in our analysis, so that the base of the logarithm in the local skew bound becomes $\sigma=\mu/(\zeta-1)$.

Ideally, we want to choose $\zeta$ as close as possible to $\vartheta$ without breaking the crucial invariant that $v_0$ never satisfies the fast condition on a level $s>s_0$.
To achieve this, it suffices to maintain that $\Psi_{v_0}^s(t)=0$.
In the following, denote by $H=(V',E')=(V\dot{\cup} \{v_0\},E\dot{\cup}\{\{v_0,v\}\,|\,v\in R\})$ the simulated graph.

\subsection*{Basic Algorithm}
For the sake of exposition, we first assume that the algorithm runs on $H$ without the ``slowdown'' of clocks discussed above, where $v_0$ increases its logical clock at rate $\zeta$.
At the end of this subsection, we derive our results by adding the simulation of $v_0$, which requires dividing all clock rates by $\zeta$.
For notational convenience, throughout this subsection we do not distinguish $H$ from $G$ in the notation, with the understanding that all statements are about $H$ and the associated values, not $G$.

First, we re-establish the key lemmas of \Cref{sec:algo} in slightly modified form, taking into account the behavior of $v_0$.
\begin{lemma}\label{lem:virtual_waitup}
Let $s_0< s\in \N$, $v\in V$, and $t'>t$.
If $\Psi_v^s(\tau)>0$ for all $\tau\in(t,t')$, then
\begin{equation*}
\Psi_v^s(t')\le \Psi_v^s(t)+\zeta(t'-t)-(L_v(t')-L_v(t)).
\end{equation*}
\end{lemma}
\begin{proof}
The proof is analogous to that of \Cref{lem:waitup}, with the following exception.
If $w=v_0$ maximizes $\Psi^s_v(\tau)$, we have that $\frac{dL_w}{dt}(\tau)=\zeta$, which is reflected by $\vartheta$ being replaced by $\max\{\vartheta,\zeta\}=\zeta$ in the claim of the lemma.
\end{proof}
For $v=v_0$ we get the stronger statement that $\Psi_{v_0}^s$ decreases at rate at least $\zeta-\vartheta$ whenever positive.
\begin{lemma}\label{lem:virtual_waitup_v0}
Let $s_0< s\in \N$ and $\Psi_{v_0}^s(t)>0$.
Then
\begin{equation*}
\frac{d\Psi_{v_0}^s}{dt}(t)\le \vartheta-\zeta.
\end{equation*}
\end{lemma}
\begin{proof}
As $\Psi_{v_0}^s(t)>0=L_{v_0}(\tau)-L_{v_0}(\tau)-d^s(v_0,v_0)$, $v_0$ does not maximize $\Psi_{v_0}^s(t)$.
Thus, analogous to the proof of \Cref{lem:waitup}, we obtain that
\begin{equation*}
\frac{d\Psi_{v_0}^s}{dt}(t)\le \vartheta-\frac{dL_{v_0}}{dt}(t)=\vartheta-\zeta.\qedhere
\end{equation*}
\end{proof}
\begin{corollary}\label{cor:virtual_waitup_v0}
Let $s_0< s\in \N$. Then for $t\ge a+\Psi^s(a)/(\zeta-\vartheta)$, it holds that $\Psi_{v_0}^s(t)=0$.
\end{corollary}
\begin{proof}
By \Cref{lem:virtual_waitup_v0}, $\Psi_{v_0}^s$ becomes $0$ within $\Psi^s_{v_0}(a)/(\zeta-\vartheta)\le \Psi^s(a)/(\zeta-\vartheta)$ time.
As it is non-negative (since $L_{v_0}(t)-L_{v_0}(t)-d^s(v_0,v_0)=0$), the lemma implies that this becomes an invariant once this happens.
\end{proof}

\Cref{lem:catchup} applies under the additional condition that $\Psi_{v_0}^{s-1/2}(t)=0$.
\begin{lemma}\label{lem:virtual_catchup}
For each $s_0<s\in \N$, node $v\in V$, and times $t<t'$, we have the following.
If $\Psi_{v_0}^{s-1/2}(t)=0$, then
\begin{equation*}
L_v(t')-L_v(t) \ge t'-t+\min\left\{\Psi^{s-1/2}_v(t),\mu(t'-t)-\Psi^{s-1/2}(t)+\Psi^{s-1/2}_v(t)\right\}.
\end{equation*}
\end{lemma}
\begin{proof}
If $v=v_0$, the claim vacuously holds, since $L_{v_0}(t')-L_{v_0}(t)>t'-t=t'-t+\Psi_{v_0}^s(t)$.
If $v\neq v_0$, the proof is analogous to that of \Cref{lem:catchup}, with the following addition showing that the behavior of $v_0$ does not interfere with the reasoning.
Since $\Psi_{v_0}^{s-1/2}(t)=0$, we have that
\begin{equation*}
\phi_{v_0}(t)=L_w(t)-L_{v_0}(t)-d^{s-1/2}(v_0,w)\le \Psi_{v_0}^{s-1/2}(t)=0.
\end{equation*}
As shown in the proof, this entails that $\phi_{v_0}(\tau)\le 0$ for all $\tau\ge t$.
Therefore, $v_0$ can never maximize $\Phi(\tau)>0$ for any $\tau \ge t$.
\end{proof}
With these modified statements, our machinery yields the following result.
\begin{corollary}\label{cor:virtual_psi}
Suppose that $s_0<s\in \N$, $\sigma:=\mu/(\zeta-1)\ge 2$, $\zeta\ge 2\vartheta-1$, $\varepsilon>0$ is constant, and $T\ge C(\Psi^s(a)+\W^s)/(\zeta-1)$ for a suitable constant $C>0$.
Then for each $s\le s'\in \N$ and time $t \in [a+T/2,a+T]$, it holds that
\begin{align*}
\Psi^{s'}(t)&\le \frac{(2+\varepsilon)\W^s}{\sigma^{s'-s}(\sigma-1)},\\
\G(t) &\le \W^s+\frac{(2+\varepsilon)\W^s}{\sigma-1},\mbox{ and}\\
\Ll_{\{v,w\}}(t) &\in |O_{v,w}|+4\left(s+\log_{\sigma}\left(\frac{\W^s}{\delta_{\{v,w\}}}\right)+\BO(1)\right)\delta_{\{v,w\}}\mbox{ for each }\{v,w\}\in E. 
\end{align*}
\end{corollary}
\begin{proof}
By \Cref{cor:virtual_waitup_v0}, we have that $\Psi_{v_0}^s(t)=0$ for times $t\ge a+\Psi^s(a)/(\zeta-\vartheta)$.
Since $\zeta\ge 1+2(\vartheta-1)$ and $\Psi^s(a)\in \BO(\W^s)$, this implies that $\Psi_{v_0}^s(t)=0$ at times $t\ge a+\BO(\W^s/(\zeta-1))$.
The corollary then is shown analogous to \Cref{sec:algo}, with $\zeta$ replacing $\vartheta$ and \Cref{lem:virtual_waitup,lem:virtual_catchup} taking the role of \Cref{lem:waitup,lem:catchup}, respectively.
\end{proof}

\subsection*{Bounding Stabilization Time}
Conceptually, the machinery from \Cref{sec:stab} carries over without significant change.
However, note that the virtual node $v_0$ taking part in the tree and the flooding would have to be simulated using the actual graph, i.e., all nodes in $R$ would have to communicate with each other, regardless of how far apart they are in $G$.
This has the consequence that $\W^s$, not its equivalent $\W_H^s$ in $H$, appears in the respective statements.
Moreover, we cannot reset the logical clock of $v_0$, as this would equate to changing the reference time.

We now briefly discuss how to handle $v_0$ in the lemmas from \Cref{sec:stab}.
\begin{itemize}
  \item \Cref{lem:bfs} is not affected, since we continue to operate on a tree in $G$.
  \item As the flooding does not include $v_0$, there is no time $t_{v_0}$ and the root obtains no values $o_{v_0,v}(t_{v_0})$, $v\in R$. However, for each $v\in R$ we can set $o_{v_0,v}(t_v):=-o_{v,v_0}(t_v)$, which satisfies the required error bounds on $o_{v_0,v}$. Moreover, we know that $L_{v_0}(t_v)=t_v$ by construction. For any walk in $H$ leaving $v_0$ via edge $\{v_0,v\}$, we then use $t_v$ in lieu of $t_{v_0}$ in the estimation done in the proof of \Cref{lem:est_Psi}. Thus, \Cref{lem:est_Psi} generalizes to $H$.
  \item For \Cref{lem:reset}, the root now computes an assignment satisfying that $\ell_{v_0}=0$, in turn dropping the requirement that $\ell_r=0$. The proof then proceeds analogously, noting that the virtual node $v_0$ does not need to adjust its clock.
\end{itemize}
A generalization of \Cref{thm:stab} to $H$ is now shown analogously, where \Cref{cor:virtual_psi} takes the place of \Cref{cor:psi}.
Since $L_{v_0}(t)=t$, the global skew bound turns into a bound on the real-time skew~$T$.
\begin{theorem}\label{thm:virtual_stab}
Suppose that $R\neq \emptyset$, $\sigma:=\mu/(\zeta-1)\ge 2$, $1+2(\vartheta-1)\le \zeta<1+\mu$, $\varepsilon>0$ is constant, and $T\ge C\W^s/(\zeta-1)$ for a suitable constant $C>0$.
Then a variant of \Cref{alg:gcs_basic} satisfies the following.
There is $\tilde{s}_0\in \N$, $s_0<\tilde{s}_0\in s_0+\BO(1)$, such that for each $\tilde{s}_0\le s\in \N$ and time $t \in [a+T/2,a+T]$, it holds that
\begin{align*}
\Psi^s(t)&\le \frac{(2+\varepsilon)\W^{\tilde{s}_0}_H}{\sigma^{s-\tilde{s}_0}(\sigma-1)},\\
\T(t)\le \G(t) &\le \W_H^{\tilde{s}_0}+\frac{(2+\varepsilon)\W^{\tilde{s}_0}_H}{\sigma-1},\mbox{ and}\\
\Ll_{\{v,w\}}(t) &\in |O_{v,w}|+4\left(\tilde{s}_0+\log_{\sigma}\left(\frac{\W^{\tilde{s}_0}_H}{\delta_{\{v,w\}}}\right)+\BO(1)\right)\delta_{\{v,w\}}\mbox{ for each }\{v,w\}\in E.\\
\end{align*}
The output clocks of the algorithm satisfy $\alpha=1/\zeta$ and $\beta=\vartheta(1+\mu)/\zeta$.
\end{theorem}
\Cref{cor:external} follows from this theorem, analogously to \Cref{cor:uniform}.

\section{Syntonization}\label{sec:syntonization}

So far, our analysis bounded skews in terms of the worst-case frequency error of the hardware clocks, i.e., $\vartheta-1$.
In many cases, this yields needlessly weak bounds, as the frequency stability of the hardware clocks is much better.
To resolve this issue, we propose to first derive intermediate, ``frequency-stabilized'' clocks, whose frequency we tie to a single reference.
As all nodes' hardware clocks, including that of the reference, change frequency slowly, this results in a much better match of hardware clock rates throughout the network.

To this end, we use \emph{phase-locked loops} to lock each node's local oscillator to a common reference, e.g.\ the one of the root $r$ or the real time provided to nodes in $R$ by their external reference(s).
As the physics of oscillators and PLLs are non-trivial, we do not attempt to comprehensively model them here.
Instead, we assume that when providing the PLL with a control input of a given period length $P$ that deviates by factor at most $1\pm \varepsilon$ from the true period of the reference throughout a time interval $T$, there is a value $\nu(P)$ such that the locked PLL guarantees a frequency between that of the reference and $1+\nu(P)+\BO(\varepsilon)$ times that of the reference.
Here, $\nu(P)$ can be thought of the performance of the PLL when its input would be ``accurate'' w.r.t.\ the reference, while a phase error of $\varepsilon P$ in transferring the reference signal should not increase the resulting frequency error by more than $\BO(\varepsilon)$.

Essentially, this can be read as $\varepsilon$ being small enough for the PLL to lock to the reference frequency, i.e., function as intended.
Since the goal is to improve the performance of a free-running oscillator, $\varepsilon<\vartheta-1$ is going to be small enough for this assumption to be valid.

What we seek to show in this section is that we can control $\varepsilon$ as $W^{\tilde{s}_0}/P$.
We then briefly discuss that $P$ can be expected to be large enough for this to improve performance dramatically, i.e., allowing us to replace $\vartheta$ by $1+\BO(\nu(P)+\varepsilon)$, where the reference then becomes the source of the ``real'' time $t$.
The only price to pay is that the stabilization time increases by the time for a PLL with period $P$ of the control signal to lock.

Showing this based on our preceding analysis is straightforward.
Recall that, as a by-product of our approach to stabilizing the GCS algorithm, every $\Theta(S+\W^{\tilde{s}_0}/\mu)$ time the root $r$ of the tree estimates $L_v(t_r)-L_r(t_r)$ for each $v,w\in V$.
\begin{corollary}\label{cor:synt_L}
When computing $\psi$ as in \Cref{cor:est_Psi}, the root $r$ can also estimate $L_v(t_r)-L_w(t_r)$ for each $v,w\in V$ with error $\BO(\W^s)$ for some $s_0<s\in s_0+\BO(1)$.
\end{corollary}
\begin{proof}
In the proof of \Cref{lem:est_Psi}, $\psi$ is obtained by taking the maximum over all pairs $(v,w)\in V\times V$ of their negative distances $-d(v,w)$ in $(V,\vec{E},4(s_0+c+2)\delta_e-\omega_e)$, for which it is shown that
\begin{equation*}
L_w(t_r)-L_v(t_r)-d^{s_0+1}(v,w)-2\Wdelta\le -d(v,w)\le L_w(t_r)-L_v(t_r)-d^s(v,w)+2\Wdelta,
\end{equation*}
where $s=s_0+2c+3$.
Since $2\Wdelta\le \W^s$ by \Cref{lem:Wdelta} and $d^{s_0+1}(v,w)\le d^s(v,w)\le \W^s$, the claim follows.
\end{proof}
The same procedure can be utilized to estimate offsets of hardware clocks, with the same error bound.\footnote{A more careful analysis would exploit that for hardware clocks, $\beta-\alpha=\vartheta-1$, yielding a stronger bound. For simplicity, we avoid introducing additional notation.}
\begin{corollary}\label{cor:synt_H}
One can modify the procedure in \Cref{sec:stab}, so that when computing $\psi$ as in \Cref{cor:est_Psi}, the root $r$ can also estimate $H_v(t_r)-H_w(t_r)$ for each $v,w\in V$ with error $\BO(\W^s)$ for some $s_0<s\in s_0+\BO(1)$.
\end{corollary}
\begin{proof}
Formally, we may run a trivial algorithm that simply outputs $H_v(t)$ when queried at time $t$.
The claim therefore follows from \Cref{cor:synt_L}.
\end{proof}
\begin{proof}[Proof of \Cref{thm:frequency}]
Each node $v\in V$ runs a PLL, whose control input is provided by the offset measurements the root can compute according to \Cref{cor:synt_H}.
Here, the root ensures that the measurements are initiated every $P$ time according to its hardware clock, and the results are distributed to each $v\in V$ via the tree.\footnote{We assume that the PLL design takes into account that the measurement results are available only $\Wd$ time after $t_r$. As $P\ge \Wd$, this has no asymptotic effect.}

When synchronizing externally, $r$ needs to simulate this procedure with $v_0$ as the root, since we cannot (and must not) change the frequency of the external references.
In this case, $r$ does not know $t_{v_0}$ and uses $t_r$ as a proxy.
As $t_r-t_{v_0}\le \Wd$, analogous reasoning to earlier bounds shows that this does not affect the asymptotic guarantees.
\end{proof}

\section*{Acknowledgements}
The author would like to thank the anyonymous reviewers for their valuable feedback.

\bibliographystyle{plain}
\bibliography{bibliography}

\appendix

\section{Discussion of the Model}\label{app:discussion}
The purpose of this appendix is to comprehensively justify the model choices we made.
In many cases, these involve crucial quality-of-live improvements to keep the analysis manageable;
however, we also put particular emphasis on the novel model extensions introduced in this work.

\subsection*{Clocks}
In practice, an algorithm in the kind of systems we model cannot read hardware or output clocks at arbitrary times.
Moreover, modeling clocks at the physical level would necessitate a stochastic description reflecting their short-term noise behavior.
Fortunately, these issues can be largely abstracted away as follows.
First, as we do not prescribe when exactly computational steps are taken and we assume bounds on communication delays only, we can pretend that they are taking place when clocks \emph{are} read.
This provides enough flexibility so that low-level hardware implementations can match the high-level description of algorithms used in this work, yet the assumption of continuous, differentiable clocks is valid.
We remark that the latter is more than mathematically required---one could work with the Lipschitz condition imposed by the rate bounds---but it avoids cumbersome, yet ultimately uninformative technical complications in proving the results.

The use of unbounded, real-valued clocks is primarily a matter of notational convenience, with the understanding that a real system uses bounded, discretely-valued clocks.
As already pointed out, the abstraction of continuous clocks can be justified.
To address that clock values must not grow without bound, a ``wrap-around'' needs to be implemented, i.e., the relevant registers eventually overflow.
Noting that the algorithm's goal must be to ensure some global skew bound $\G$ at all times, using clocks that wrap around after $M\gg \G$ time enables to consistently interpret offset measurements as if the clock values were unbounded;
effectively, the nodes compute and output the hypothetical unbounded output clocks modulo $M$.
It should be stressed that self-stabilization requires making sure that the state is consistent, i.e., clock values are indeed within a sufficiently small range, and the mechanism for offset estimation and recovering a small global skew must not be mutually dependent on each other.
However, this is straightforward by triggering a reset whenever excessive skews are detected, so we chose to avoid the significant notational clutter resulting from including such details.

While the above modeling choices have already been made in previous work, the inclusion of the aspect of their frequency stability is largely novel in theoretical work on distributed clock synchronization.
Khanchandani and Lenzen touch on this issue by assuming a bounded second derivative of $H_v$ in~\cite{khanchandani19precision}.
However, oscillators are physically highly complex, and their timescale-dependent noise characteristics are best modeled as stochastic processes. These behaviors are not understood from first principles and must instead be empirically characterized; see~\cite{rubiola2009phase} for a comprehensive discussion.
As a result, such a modeling approach might not adequately predict the performance of an algorithm.
Instead, in \Cref{sec:syntonization} we fell back to highlighting the issue that it is feasible to separate the concern of obtaining stable local reference clocks from the issue of minimizing clock offsets, essentially reproducing the performance of PLLs in terms of frequency stability.

Nonetheles, we stress that care should be taken to ensure that insufficient frequency stability of (locked) oscillators does not become a performance bottleneck.
PLLs, especially when locking all local oscillators to the \emph{same} frequency reference, will tie them to the performance of the reference only beyond the frequency at which the PLLs operate.
On timescales that are smaller, but still relevant to the algorithm, i.e., larger than the update rate of measurements, the local oscillators' characteristics are crucial.
For example, back-of-the-envelope calculations for a VLSI system suggest that the frequency stability of cheap quartz oscillators might \emph{just} be enough when locking at the empirically determined sweet spot of around 10--100\,kHz.
However, depending on the numbers, it might as well be advisable to slightly improve the characteristics for better results.
Similarly, when synchronizing base stations wirelessly, the challenges imposed by radio communication might result in large communication delays and hence push this transition point further out, which in turn requires more stable oscillators to attain the feasible performance.

\subsection*{Offset Estimates}
The choice to disallow parallel edges has been made to avoid notational complications.
The same techniques and arguments can be applied when allowing parallel edges.
However, from the perspective of a worst-case analysis, there is very limited gain, as an accurate parallel edge cannot ``make up'' for a highly inaccurate edge beyond cutting the error roughly in half.
Moreover, such an approach could instead be modeled by creating a single, ``synthesized'' offset estimate out of the potentially multiple methods at the disposal of designers.

Next, we note that we assumed that only minimal knowledge about upper bounds on estimate errors is available, i.e., just enough to choose $T$ sufficiently large for the constraint that $T\ge C\W^s/\mu$ or $T\ge C\W^s/(\zeta-1)$ to be met, respectively.
One could hope that having more information could enable stronger results.
However, in light of the example given in the introduction, this seems questionable in scenarios where $\delta(T)\ll \Delta$, which is what we consider in this work.
Moreover, there is a strong incentive not to rely on known upper bounds:
as a result, up to the dependence on $\delta(T)$, the skew bounds we obtain depend on the current ``true'' errors only, i.e., the performance of the algorithm closely tracks the performance of the estimates at any given time.
This is likely to result in much better bounds than plugging in a worst-case bound that holds across all times and deployed systems.

Recall that we required that offset estimation errors are antisymmetric up to variations bounded by $\delta(T)$, within a time window of size $T$.
This is easily realized from measurements without this property by having $v$ and $w$ exchange their estimates and average, for example.
This comes at the ``cost'' of an additional message delay, but any functional offset estimation method will involve communication in the first place, so we do not view this as a downside.
Nonetheless, one could drop this requirement, instead falling back to treating the forward and backward direction of an edge as closing a cycle.
This is perfectly fine in terms of the analysis, but would result in a pointless performance degradation.
That is illustrated best by considering trees:
here, this approach artificially introduces negative-weight cycles, degrading the shown result relative to being able to work with $s=0$.
Instead, we insisted on avoiding this issue, both for increased clarity when comparing to prior work on GCS and to emphasize that we obtain results that strictly improve on state-of-the-art tree based methods that are deployed in practice.

Another major point is that we do not answer the question how clock offset estimates are obtained.
This has two key reasons:
first, the point is that the approach works with \emph{any} suitable method, and hence can be seen as reducing the synchronization task to obtaining the best possible local clock offset estimates;
second, a comprehensive treatment is entirely impractical within the scope of this work.
To illustrate the latter point, let us list a few scenarios that all require different techniques:
\begin{compactitem}
\item The most standard scenario considered would be a static, wired network. Here, one typically measures using a two-way exchange, cancelling out any symmetric variation in the delays.
\item If nodes are mobile with relatively large speeds, speeds and acceleration information should be considered as well.
\item In wireless settings, one might need long measurements using carefully constructed radio codes, to enable accurate measurements despite signal strengths below the noise floor. This then involves additional challenges, e.g.\ that one needs to correct for the frequency error between the oscillators driving the sender's and the receiver's radio.
\item With very high speeds or when using slow-moving signals such as sound (e.g.\ in underwater networks), one might need to compensate for the Doppler effect.
\item Kuhn and Oshman~\cite{kuhn09reference} propose to employ reference broadcasts, where the same transmission is time-stamped by multiple receivers to obtain accurate offset measurements.\footnote{They argue that arrival times of the signal are essentially the same due to the large speed of light. For nanosecond accuracy, even in networks covering a very small area one would have to include distance information, however: a radio signal will travel by no more than about $30$\,cm in a nanosecond.}
\end{compactitem}
This list is by no means exhaustive, but should illustrate the point made above.
Without claiming completeness, we refer the interested reader to relevant work covering the above points~\cite{ptp,etzlinger18synchronization,ni21underwater}.
However, we would like to stress that the frequency error of the local oscillators---in practice after locking them to a frequency reference as described in \Cref{sec:syntonization}---needs to be considered as a crucial source of error in the offset measurements.
The best estimation method will be useful only up to the point where the time for obtaining such a measurement becomes so large that the error accumulated due to the frequency differences between the output clocks starts to dominate.
Note that this might crucially affect the choice of $\mu$, as it is the frequency error of the output clocks that is decisive.
Thus, our work carries the additional advantage that we can choose $\mu$ small due to the reduced impact of the logarithmic term on the local skew, i.e., it becomes less relevant to choose $\mu$ large in order to ensure large $\sigma$.

Last but not least, we would like to point out that we deliberately left out assumptions on probabilistic behavior of delays, clocks, etc.
First, this is simply a matter of tractability, as such assumptions would significantly complicate analysis and, possibly, the algorithmic strategies required, putting it firmly out of scope of the main question driving this work.
Second, much of such behavior can---and arguably \emph{should}---be contained to how clock offset measurements are obtained and maintained, as well as determining the best strategies for providing hardware clocks $H_v$ of small frequency error.
Nonetheless, it remains an important and open challenge to study the average case performance of GCS algorithms and whether it can be improved without sacrificing worst-case guarantees in relevant application scenarios.

\subsection*{Communication}

To simplify the obtained bounds and streamline the presentation, we assumed that estimates are ``good enough'' relative to the delays involved in explicit communication.
In general, this is not a given.
For instance, there might be an accurate measurement routine that is provided, but additional communication---be it due to bandwidth concerns or system layer separation---might use a different mechanism or be less privileged.
However, little insight can be conveyed by generalizing the presented bounds using additional parameters covering this, and it should be straightforward to estimate the effect on performance by examining where the respective assumption is used.
In this regard, it might be more important to ensure that the technique we propose for self-stabilization is refined to avoid large communication bandwidth requirements.
Here, we consider a Bellman-Ford style approach to aggregating the required information to be very promising, but leave exploring this to future work.

\subsection*{The Effect of Requiring Convergence within $T$ Time}
While not an effect of modeling choices in the strict sense, we opted for analyzing the GCS algorithm within a time window of duration $T$ only, deriving general guarantees by splitting the time axis into overlapping time segments of duration $T$.
This reduces the notational burden and simplifies formal statements, but comes at the disadvantage of necessitating that $T\ge C\W^s/\mu$ or $T\ge C\W^s/(\zeta-1)$, respectively.

At closer inspection, one could argue that potentials decrease sufficiently within $T$ time such that their possible increase caused by adjusting the $O_{v,w}$ values when moving to a new time segment is canceling out at most half of the progress made.
This would be good enough to guarantee the same asymptotic bounds on stabilization time, while replacing the $\W^s$ factor in the constraint by the amount by which $\W^s$ \emph{changes} within $T$ time.
However, we could not clearly identify a scenario in which this offers substantial gain:
once $\delta(T)$ becomes small enough for the logarithmic term to not dominate the bounds, reducing $\delta(T)$ further makes the algorithm less robust (as a stronger guarantee must be met), yet offers diminishing returns.
In other words, it is not crucial to minimize $T$, but rather sufficient to reduce it to the point where the contributions to skews arising from slower effects become dominant.
In light of this, we refrained from attempting to formalize the above insight.

\end{document}